\documentclass[paper=letter, fontsize=12pt]{article} 
\usepackage{mystyle}
\usepackage{amsmath}
\usepackage{times}
\usepackage{keyval}
\usepackage{calc}
\usepackage{color}

\usepackage[letterpaper]{geometry}
\newcommand{\revised}[1]{\textcolor{black}{#1}}
\newcommand{\drevised}[1]{\textcolor{black}{#1}}

\renewcommand{\baselinestretch}{1.1}

\usepackage{fancyhdr}
\usepackage{pslatex}

\usepackage[round]{natbib}

\begin{document} 
\rhead{\textit{December 11, 2014}}
\lhead{\textit{C.B.\ Hyndman \& X.\ Zhou}}
\chead{\textit{Explicit solutions of quadratic FBSDEs in QTSMs}}

\title{
Explicit solutions of quadratic FBSDEs arising from quadratic term structure models
}

\author{
Cody Hyndman\footnote{Corresponding Author: email: cody.hyndman@concordia.ca}\ \footnote{ 
Department of Mathematics and Statistics, 
Concordia University, 
1455 Boulevard de Maisonneuve Ouest,
Montr\'eal, Qu\'ebec,
Canada H3G 1M8.
}
\ and 
Xinghua Zhou\footnote{
Department of Applied Mathematics,
Western University, London, Ontario, Canada, N6A 5B7
}
}

\date{December 11, 2014}

\maketitle

\abstract{
We provide explicit solutions of certain forward-backward stochastic differential equations (FBSDEs) with quadratic growth.  These particular FBSDEs are associated with quadratic term structure models of interest rates and characterize the zero-coupon bond price.  The results of this paper are naturally related to similar results on affine term structure models of Hyndman~(Math. Financ. Econ. 2(2):107-128, 2009) due to the relationship between quadratic functionals of Gaussian processes and linear functionals of affine processes.  Similar to the affine case  a sufficient condition for the explicit solutions to hold is the solvability in a fixed interval of Riccati-type ordinary differential equations.  However, in contrast to the affine case, these Riccati equations are easily  associated with those occurring in linear-quadratic control problems.  We also consider quadratic models for a risky asset price and characterize the futures price and forward price of the asset in terms of similar FBSDEs.  An example is considered, using an approach based on stochastic flows that is related to the FBSDE approach,  to further emphasize the parallels between the affine and quadratic models.  An appendix discusses solvability and explicit solutions of the Riccati equations.
}

\noindent
\textbf{Keywords:}
Quadratic term-structure models;  forward-backward stochastic differential equations; zero coupon bond price; quadratic price model; futures price; forward price; Riccati equations.

\vspace{5mm}
\noindent
\textbf{JEL Classification:} E43, G12, G13 

\vspace{5mm}
\noindent
\textbf{Mathematics Subject Classification (2000):} 60G35, 60H20, 60H30, 91B28, 91B70

\vfill
\pagebreak

\renewcommand{\baselinestretch}{1.5}

\section{Introduction}

An important class of term-structure models are  affine
term-structure models (ATSMs).  The defining characteristic of an
ATSM is that the price at time $t\in[0,T]$ of a unit face value $T$-maturity
zero-coupon bond, denoted  by $P(t,T)$,  is an  exponential-affine function of
an $n$-dimensional factor process $X_t$.  That is, for times $0\leq t\leq T$
\begin{equation*}
P(t,T,X_t)=\exp{\left(B(t,T)\transpose X_t+C(t,T)\right)}
\end{equation*} 
where $B(t,T)$ is an $n\times 1$ vector and $C(t,T)$ is a scalar.  
As a result, the yield of the bond is an affine function of
the factor process.  The class of ATSMs includes the models of
\citet{Vas77}, \citet*{CIR}, \citet{Duffee}, \citet{MR1994043}, and many others.
Despite several attractive properties ATSMs have been 
demonstrated to have some empirical
limitations.  For example, \citet{Dai} show that ATSMs fail to
capture certain aspects of the distribution of swap yields,
suggesting ATSMs may omit empirically observed nonlinearities. 
Further, \citet{AhnDong} demonstrate empirically that non-affine term structure models
outperform one-factor affine models.

In order to address the limitations of ATSMs several authors have
proposed the use of quadratic term-structure models (QTSMs). 
In a QTSM zero-coupon bond prices are exponential-quadratic functions 
of the factor process $X_t$ for times $0\leq t\leq T$
\begin{equation*}
P(t,T,X_t)=\exp{ \left( X_t\transpose A(t,T) X_t+  B(t,T)\transpose
X_t+C(t,T)\right) }
\end{equation*}
 where $A(t,T)$ is a non-singular $n\times n$
matrix, $B(t,T)$ is an $n\times 1$ vector, and $C(t,T)$ is a scalar.
\citet{QuadraticTheoryEvidence} introduce the comprehensive QTSMs and
study the characteristics of these models. Pricing problems associated with
QTSMs have been studied by \citet{QuadraticTermStructure} and
\citet{Markus}. Other relevant research on QTSMs includes
\citet{Pseudodiffusions} and \citet{Eigen} which provide further evidence that
 QTSMs can capture nonlinearities between economic factors and provide more
flexibility when constructing models when compared to ATSMs. 
Moreover, as shown by
\citet{QuadraticTermStructure} and \citet{Markus,Markus2}, QTSMs are analytically tractable as the prices of European style options can be calculated by Fourier
transform methods similar to ATSMs.  \citet{raquel} also considers quadratic term structures for bond, futures, and forward prices.

In this paper we consider  QTSMs using two nontraditional, but related, approaches to pricing problems.  The first approach, and our main focus, is based on forward-backward stochastic differential equations (FBSDEs), which we henceforth refer to the \textit{FBSDE approach} and was previously introduced in the context of ATSMs in \cite{thesis,Cody2007CIR,Cody2009}.  By first characterizing the factor process and bond price in terms of the solution of coupled nonlinear FBSDEs and then demonstrating existence, uniqueness, and explicit solutions of the the FBSDEs the pricing problem is solved. The key result of the FBSDE approach is the extension of a technique due to 
\citet{Yong} of a method for proving the existence, uniqueness, and explicit solution of certain coupled linear FBSDEs to the nonlinear FBSDEs which characterize the bond pricing problem in ATSMs.  The same techniques were employed to characterize futures prices and forward prices in affine price models (APMs) in \cite{Cody2009}.  In this paper we extend the FBSDE approach to the bond pricing problem in the context of QTSMs and futures and forward prices of quadratic price models (QPMs) for a risky asset.  We \revised{obtain} results characterizing these \revised{prices} which are similar to the ATSM case and in particular provide new examples of quadratic FBSDEs with explicit solutions.

The second approach, which we briefly consider, is based on the stochastic flows method 
studied by \citet{Elliott}, \citet{Martino}, and \citet{hyndman:Gauss,Cody2009}. This method gives a closed-form solution to the
pricing problems for certain ATSMs.  \citet{Geman} and \citet{Yor} have shown that the CIR process is
a Bessel process under certain restrictions, which means that the
CIR process and QTSMs are equivalent in certain
cases. Motivated by this fact, we extend the techniques of the
stochastic flows approach and the FBSDE approach to QTSMs.

The paper is organized as follows. Section~\ref{Sec:1} briefly introduces the
modelling framework and notation.  Section~\ref{sec:2} reviews the  FBSDE approach 
for the zero-coupon bond price and extends the results of \citet{Cody2009} from ATSMs to QTSMs.  Section~\ref{sec:fut-fwd} considers models where a risky asset price is an exponential quadratic function of the factor process (QPM), which includes the zero-coupon bond price, and applies the FBSDE approach to the futures price and forward price. Section~\ref{sec:flows} briefly considers the stochastic flows approach to QTSMs.  We give an explicit solution for the zero-coupon bond price \revised{for} our model in the one-dimensional case based on the flows method. Section~\ref{sec:conc} concludes and the Appendix considers the solvability of certain matrix Riccati-type differential equations which \revised{give sufficient} conditions for the main results of this paper.

\section{Preliminaries and notation}\label{Sec:1}

We shall begin our analysis on the risk-neutral filtered probability space 
$\filtprobspace{F}{}{Q}$ for $0\leq t\leq T^*$ where $T^*$ is the fixed and finite investment horizon,
\{$\filtration{t}{}\}$ is a right-continuous and complete filtration satisfying the usual conditions,
and $Q$ is the risk-neutral (martingale) measure. Under these assumptions, as in \citet[p. 411]{Steve}, the
price of the zero-coupon bond at time $t$ for maturity $T\leq T^*$
is given by
\begin{equation}\label{BondPrice}
P(t,T)= \condexpectation{Q}{ \exp{\left( -\int_{t}^{T} r_u \ du \right) }}{\filtration{t}{}}
\end{equation}
where $r_{t}$ is the  instantaneous riskless interest rate.
It is possible to calculate the conditional expectation~(\ref{BondPrice}) in several different ways after specifying the dynamics of the riskless interest rate.

On the risk-neutral probability space $\filtprobspace{F}{}{Q}$ assume that the
riskless interest rate is a function of an $\real^{n}$-valued, 
$\{\filtration{t}{}\}$-adapted state process $X_{t}$ given as the solution to the Gaussian 
stochastic differential equation (SDE)
\begin{equation}\label{FactorModel}
dX_t= (AX_t+B)dt+\sigma dW_t, \quad X_{0}=x_{0}
\end{equation} with $A = [a_{i,j}]$ an $(n\times n)$-matrix, $B=[\tilde{b}_{i}]\transpose$ an $(n\times 1)$-column vector, $\sigma=[\sigma_{i,j}]$ an $(n\times n)$-matrix, and $W_t = [W_t^{(1)} , \cdots , W_t^{(n)}]\transpose$ a standard Brownian motion with respect to $(\filtration{t}{},Q)$.
We assume that the riskless interest rate $r_t$ is given by a quadratic function of the factors process.
\begin{assumption} \label{ModelA}
The instantaneous riskless interest rate process is given by $r_{t} = r(X_{t})$ where $X_{t}$ is the strong solution to equation~(\ref{FactorModel}) and, for $x\in\real^{n}$,
\begin{equation}
r(x)=x\transpose \Gamma x + R x + k, \label{eq:ModelA}
\end{equation} 
where $\Gamma = [\gamma_{i,j}]$ is a positive semidefinite $(n\times n)$-matrix, $R = [r_{1}, \ldots, r_{n}]$ is an $(1\times n)$-row vector, and $k$ is a scalar such that
\begin{equation}
k \geq \frac{1}{4}R\Gamma^{-1}R\transpose . \label{eq:pos_cond}
\end{equation}

\end{assumption}
One of the common criticisms of Gaussian short rate models, such as the \cite{Vas77} model, is the potential for producing negative interest rates.  However, in this case since $\Gamma$ is positive semidefinite the lower bound of $r(x)$
from equation~(\ref{eq:ModelA}) is $(k-\frac{1}{4}R\Gamma^{-1}R\transpose)$
when $x=-\frac{1}{2}\Gamma^{-1}R\transpose$.
Therefore, the model given by Assumption~\ref{ModelA}, with the restriction of equation~(\ref{eq:pos_cond}), 
produces a nonnegative instantaneous interest rate process $r_{t}=r(X_{t})$ for $t\geq 0$.

Assumption~\ref{ModelA} and equation~(\ref{BondPrice}) may now be used to extend the forward-backward stochastic differential equation (FBSDE) approach to term structure modelling for ATSMs of \citet{Cody2009} to QTSMs.  The main difference from \cite{Cody2009} is that in this paper the dynamics of the factor process are given by a Gaussian (rather than affine) process and the
riskless interest rate is a quadratic (rather than affine) functional of the factors process.  
The extension is motivated by the fact that the sum of squared components of an $n$-dimensional Ornstein-Uhlenbeck process can be identified with a CIR process as shown in \citet[pp.\ 271-273]{MR2098795} and more generally by the relationships between Brownian motion and squared Bessel (BESQ) processes described by \citet{Geman} and \citet{Yor}.

\section{Connections between QTSMs and FBSDEs} \label{sec:2}
We briefly review the derivation of the forward-backward stochastic differential equation which characterizes the factor process and the bond price by considering the processes
\begin{equation*}
H_s =\exp{\left(-\int_{0}^{s}r(X_u)\ du\right)} \label{H}
\end{equation*}
and  
\begin{equation}
V_s
=\condexpectation{Q}{\exp{\left(-\int_{0}^{T}r(X_u)\ du\right)}}{\filtration{s}{}}
\label{V}
\end{equation}
for $s\in[0,T]$.  Assumption~\ref{ModelA} is in force throughout, however, provided the process $V_{s}$ defined by equation~(\ref{V}) is a $(\filtration{t}{},Q)$-martingale the characterization is valid  more generally whatever the dynamics of the factors process or functional dependence of the risk-free rate on the factors. 

Since $H_{s}$ is $\filtration{s}{}$-measurable from equation~(\ref{BondPrice}) it is clear that $P(s,T) = (V_{s}/H_{s})$. 
By the Martingale Representation Theorem  \citet[Theorem 5.4.2]{Steve} 
there exists an $\filtration{s}{}$-adapted process $J_s=[J_s^{(1)}, \cdots , J_s^{(n)}]$, expressed as an $(1\times n)$ vector process, such that
\begin{equation}
V_s=V_0+\int_{_0}^{s} J_{u} \ dW_u .  \label{MartingaleRepresentationTheoremV}
\end{equation}
Since $H_s$ is of finite variation and thus satisfies the dynamics
\begin{equation*}%
dH_s=-r(X_s)H_s \ ds .
\end{equation*}
It\^o'{s} formula gives that $Y_s= (V_s/H_s)$ satisfies
\begin{equation*}
Y_s = Y_0+\int_{0}^{s} r(X_u) Y_{u} \ du+\int_{0}^{s} Z_u  \ dW_u \label{DynamicsOfY}
\end{equation*} 
where $Z_u = ({J_u}/{H_u})$.  Subtracting the dynamics of $Y_T$ from (\ref{DynamicsOfY}) we have
\begin{equation*}
Y_s=Y_T - \int_{s}^{T} r(X_u) Y_{u} \ du - \int_{s}^{T} Z_u \ dW_u. %
\end{equation*}
Since $Y_T$ is the $T$-maturity zero-coupon bond price at time
$T$, we have from equation~(\ref{BondPrice}) that $Y_T=P(T,T)=1$. Therefore, 
on the risk neutral probability space $\filtprobspace{F}{}{Q}$, we have that for $s\in[0,T]$ the process $(X_s,Y_s,Z_{s})$ satisfies the system
\begin{align}
X_s=& X_0+\int_{0}^{s}(AX_u+B)\ du+\int_{0}^{s}\sigma \ dW_u \label{FBSDEsRN} \\
Y_s=& 1 - \int_{s}^{T} r(X_u) Y_{u} \ du-\int_{s}^{T} Z_u \ dW_u \label{FBSDEsRNb}.
\end{align}

Equations (\ref{FBSDEsRN})-(\ref{FBSDEsRNb}) constitute a
forward-backward stochastic differential equation (FBSDE).  The characterization presented demonstrates the existence and uniqueness of an adapted solution $(X_s, Y_s, Z_s)$ of the FBSDE~(\ref{FBSDEsRN})-(\ref{FBSDEsRNb}) as defined in 
\citet{Peng} (see also \revised{\citet{Karoui}} and  \citet{Ma} for a discussion of FBSDEs).  However, due to the fact that $Z_{\cdot}$ arises from an application of the Martingale Representation Theorem the solution is not known in closed form.  In order to derive a closed form solution we next apply a change of measure  and consider the dynamics of the FBSDE under the new measure.  Then, adapting a method of \citet{Yong} for the case of linear FBSDEs to the resulting nonlinear FBSDE, we are able to prove existence and uniqueness as well as provide an explicit solution.

Recall the definition of the forward measure:
\begin{definition}\label{Numeraire}
Let the zero-coupon bond be the num\'eraire.  Define
\begin{equation}
\Lambda_T  = {P(0,T)}^{-1}\exp{\left(-\int_{0}^{T} r(X_u) \ du\right)}.
\label{eq:11}
\end{equation}
Then the $T$-forward measure $Q^T$ is
defined by 
\begin{equation*}
Q^T(A):= \int_{A} \Lambda_{T} \ {dQ},
\end{equation*}
$\forall A\in\filtration{T}{}$.
\end{definition}

Define 
$\Lambda_t=E[\Lambda_T|\mathcal{F}_t]$ and note that 
\begin{align}\Lambda_t=& 
 E\bigg[P(0,T)^{-1}\text{exp}\bigg(-\int\limits_0\limits^T r(X_u) \ du
\bigg)\bigg|\mathcal{F}_t\bigg]
= V_0^{-1}E\bigg[\text{exp}\bigg(-\int\limits_0\limits^T r(X_u)\ du
\bigg)\bigg|\mathcal{F}_t\bigg]=V_0^{-1}V_t\label{Lambdat}.
\end{align}
Substitute equation (\ref{Lambdat}) into equation
(\ref{MartingaleRepresentationTheoremV}) to find
\begin{equation}
V_0\Lambda_s=V_0\Lambda_0+\int\limits_0\limits^s J_u\ dW_u\label{eq1}.
\end{equation}
Dividing both sides of (\ref{eq1}) by $V_0$ the dynamics of $\Lambda_{s}$ are given by
\begin{align}
\Lambda_s=&\Lambda_0+\int\limits_0\limits^s
\frac{J_u}{V_0}dW_u 
= 1+\int\limits_0\limits^s
\frac{J_u}{V_0}\frac{H_u}{H_u}\frac{V_u}{V_u} \ dW_u %
= 1+\int\limits_0\limits^s Y_u^{-1}Z_u\Lambda_u \ dW_u. \nonumber
\end{align}
Then, by Girsanov's Theorem, the process $\{W_{t}^{T}\}_{0\leq t \leq T}$ defined by
\begin{equation}
W^T_t = W_t -\int\limits_0\limits^t \frac{Z_u\transpose }{Y_u}du
\nonumber
\end{equation}
is a standard Brownian \revised{m}otion
under the forward measure $Q^T$. Therefore, under the forward
measure $Q^T$, the FBSDE (\ref{FBSDEsRN})-(\ref{FBSDEsRNb}) becomes
\begin{align*}
X_s &= X_0 + \int\limits_0\limits^s(AX_u+B+\sigma\frac{Z_u\transpose}{Y_u})du
    + \int\limits_0\limits^s\sigma \ dW_u^T \\ %
Y_s &= 1-\int\limits_s\limits^T \left( Y_u r(X_u)+\frac{Z_uZ_u\transpose}{Y_u}\right)du 
    - \int\limits_s\limits^TZ_u \ dW_u^T \\ %
\end{align*}
for $s\in [0,T]$.
In particular for the QTSM, with $r(X_t)$ given by (\ref{ModelA}), we have for $s\in[0,T]$
\begin{align}
X_s &= X_0 + \int\limits_0\limits^s(AX_u+B+\sigma\frac{Z_u\transpose }{Y_u})du
     +\int\limits_0\limits^s\sigma \ dW_u^T  \label{FBSDEsFWModelA}  \\
Y_s &= 1 - \int\limits_s\limits^T \left(Y_u [X_u\transpose \Gamma X_u+RX_u+k]
     +\frac{Z_uZ_u\transpose }{Y_u}\right) du
     -\int\limits_s\limits^TZ_u \ dW_u^T \label{FBSDEsFWModelAb}.
\end{align}

The FBSDE~(\ref{FBSDEsFWModelA})-(\ref{FBSDEsFWModelAb}) is nonlinear, including two quadratic terms in the driver of the BSDE, and fully coupled.  Further, given the quadratic term $(X_u\transpose \Gamma X_u)$ in the driver of the BSDE it does not fall into the class considered in \cite{Cody2009}.  According to \cite{richter:2012} there are very few examples where explicit solutions to quadratic BSDEs are available.  Obviously, the results of \citet{thesis,Cody2007CIR,Cody2009} provide examples.  We \revised{shall} prove, independent of the construction already presented, that the FBSDE~(\ref{FBSDEsFWModelA})-(\ref{FBSDEsFWModelAb}) provides another example.  Similar to the results of \citet{thesis,Cody2007CIR,Cody2009}  the solution is explicit up to the solution of a Riccati type ODE.  \revised{The Riccati equations and sufficient conditions for solvability are as follows.
\begin{lemma}\label{lemmaA}
If $\Gamma\in\real^{n\times n}$ is positive semidefinite the Riccati-type differential equations 
\begin{align} 
0 &= \frac{d}{dt} R_2(t)+(R\transpose _2(t)+R_2(t))A-\Gamma+\frac{1}{2}(R_2(t)+R_2\transpose (t))\sigma\sigma\transpose (R_2(t)+R_2\transpose (t)), & R_2(T)=0_{n\times n}
 \label{Riccati2} 
\\
0 &= \frac{d}{dt} R_1(t)+R_1(t)A+B\transpose (R_2(t)+R_2\transpose (t))-R+R_1(t)\sigma\sigma\transpose (R_2(t)+R_2\transpose (t)),
 & R_1(T)=0_{1\times n}
 \label{Riccati2b} 
\end{align}
admit unique solutions $R_2(\cdot)\in\real^{n\times n}$, $R_1(\cdot)\in\real^{1\times n}$ for all 
$t\in [0,T]$.
\end{lemma}
\begin{proof}
The nonsymmetric matrix Riccati equation~(\ref{Riccati2}) is a special case of the more general equations considered in the Appendix if we set $\Upsilon = \Gamma$ and $\Theta = 0_{n\times n}$ in equations~(\ref{eq:App_RE1})-(\ref{eq:App_RE2}) of Theorem~\ref{th:Riccati}.  Therefore, as $\Gamma$ is positive semidefinite  we have $(\Upsilon+\Upsilon\transpose)$ is positive semidefinite and $(\Theta + \Theta\transpose)$ is negative semidefinite.  That is, the conditions of Theorem~\ref{th:Riccati} are satisfied so $R_{2}(t)$ given by (\ref{Riccati2}) has a unique solution  on the interval $[0,T]$.  The solution of equation~(\ref{Riccati2b}) then follows from Corollary~\ref{CorollaryA} with $\theta = 0$ and $\Psi=R$.
\end{proof}
}
Considering the logarithm of the BSDE (\ref{FBSDEsFWModelAb}) simplifies the proof of the main result.
By It\^o's formula from $s$ to $T$ we have
\begin{equation}\label{eq2} \log Y_s =
-\int\limits_{s}\limits^T\{\frac{1}{2}\frac{Z_uZ_u\transpose}{Y_u^2}+X_u\transpose \Gamma
X_u+RX_u+k\}du - \int\limits_s\limits^T\frac{Z_u}{Y_u}dW^T_u.
\end{equation}
\begin{theorem}\label{TheoremFBA}
\revised{If $\Gamma\in\real^{n\times n}$ is positive semidefinite} the  
FBSDE~(\ref{FBSDEsFWModelA})-(\ref{FBSDEsFWModelAb}) admits a 
unique adapted solution $(X_t,Y_t,Z_t)$ for all $t\in[0,T]$ given by 
\begin{align}
X_t &=X_0+\int\limits_0\limits^t\{[A+\sigma\sigma\transpose (R_2(u)+R_2\transpose (u))]X_u+B+\sigma\sigma\transpose R_1\transpose (u)\}du+\int\limits_0\limits^t\sigma dW_u^T,\label{TheoremX}\\
Y_t &= \text{exp} \bigg( X_t\transpose  R_2(t)X_t+R_1(t) X_t+ R_0(t)\bigg),\label{TheoremY}\\
Z_t &=[X_t\transpose (R_2(t)+R_2\transpose
(t))+R_1(t)]\sigma Y_{t} %
\label{TheoremZ}
\end{align}
where \revised{$R_{2}(t)$ is the solution of equation~(\ref{Riccati2}), $R_{1}(t)$ is the solution of equation~(\ref{Riccati2b}), and}
$R_0(t)$ is given by 
\begin{equation}
\revised{R_{0}(t)} = -\int_{t}^{T}\left( k- R_1(s) B -\frac{1}{2}R_1(s)\sigma\sigma\transpose [R_1(s)]\transpose  - \trace{\sigma\transpose R_{2}(s) \sigma} \right) \ ds .
\label{TheoremR0}
\end{equation}
\end{theorem}
\begin{proof}
First, we must show that $(X,Y,Z)$ given by (\ref{TheoremX})-(\ref{TheoremZ}) satisfy the FBSDE~(\ref{FBSDEsFWModelA})-(\ref{FBSDEsFWModelAb}).  %
Dividing (\ref{TheoremZ}) by
(\ref{TheoremY}), we have\begin{equation}\label{TheoremZ/Y}
\frac{Z_t}{Y_t}=[X_t\transpose (R_2(t)+R_2\transpose (t))+R_1(t)]\sigma .
\end{equation}
Substituting (\ref{TheoremZ/Y}) into (\ref{TheoremX}), we obtain the
dynamics of $X_t$ given by equation (\ref{FBSDEsFWModelA}). So $X_t$
given by (\ref{TheoremX}) satisfies the SDE (\ref{FBSDEsFWModelA}).

Consider the function $f(t,x)=\text{exp}\bigg( x\transpose R_2(t)x+R_1(t)x+R_0(t)\bigg)$.
Applying It\^{o}'s formula to $f(t,x)$ using the dynamics of $X_t$ in (\ref{TheoremX}) we find that $Y_{t}=f(t,X_{t})$ and $Z_{t}$ defined by 
(\ref{TheoremY}) and (\ref{TheoremZ}) satisfy
\begin{align}
dY_t&=\{X_t\transpose
\dot{R}_2(t)X_t+\dot{R}_1(t)X_t+\dot{R}_0(t)+X_t\transpose
(R_2(t)+R_2\transpose (t))A
X_t\nonumber\\
&+B\transpose (R_2(t)+R_2\transpose (t))X_t+[X_t\transpose (R_2(t)+R_2\transpose (t))+R_1(t)]\sigma\frac{Z_t\transpose }{Y_t}+R_1(t)AX_t  %
+R_1(t)B
+\trace{\sigma\transpose R_{2}(t) \sigma}
\nonumber\\
&+\frac{1}{2}[X_t\transpose (R_2(t)+R_2\transpose (t))+R_1(t)]\sigma\sigma\transpose [(R_2\transpose (t)+R_2(t))X_t+R_1\transpose (t)]\}Y_tdt\nonumber\\
&+Y_t[X_t\transpose (R_2(t)+R_2\transpose (t))+R_1(t)]\sigma
dW_t^T\label{eq7}
\end{align}
where $\dot{R}_{k}(t) = \frac{d}{dt}R_{k}(t)$, $k=0,1,2$.
Substitute (\ref{TheoremZ/Y}) in (\ref{eq7}) to find
\begin{align}
dY_t&=\{X_t\transpose [\dot{R}_2(t)+(R_2(t)+R_2\transpose (t))A+\frac{1}{2}(R_2(t)+R_2\transpose (t))\sigma\sigma\transpose (R_2\transpose (t)+R_2(t))]X_t\nonumber\\
&+[\dot{R}_1(t)+B\transpose (R_2(t)+R_2\transpose (t))+R_1(t)A+R_1(t)\sigma\sigma\transpose (R_2\transpose (t)+R_2(t))]X_t\nonumber\\
&+\frac{Z_tZ_t\transpose }{Y_t^2}+\dot{R}_0(t)+R_1(t)B
+ \trace{\sigma\transpose R_{2}(s) \sigma}
+\frac{1}{2}R_1(t)\sigma\sigma\transpose R_1\transpose
(t)\}Y_tdt+Z_tdW_t^T.\label{eq8}
\end{align}
Substituting equations (\ref{Riccati2})-(\ref{Riccati2b}) and (\ref{TheoremR0}) into equation (\ref{eq8})
gives that 
$Y_t$ defined by (\ref{TheoremX})-(\ref{TheoremZ})
satisfies
\begin{equation*}
Y_t=Y_T-\int\limits_t\limits^T \{Y_u [X_u\transpose \Gamma
X_u+RX_u+k]+\frac{Z_uZ_u\transpose
}{Y_u}\}du-\int\limits_t\limits^TZ_udW_u^T.
\end{equation*} 
By the boundary conditions of (\ref{Riccati2})-(\ref{Riccati2b}) and (\ref{TheoremR0}) at $t=T$ we have, from equation (\ref{TheoremY}), that
\begin{equation*} Y_T=\exp{(X_T\transpose R_2(T)X_T+R_1(T)X_T+R_0(T) )}
=\exp{(X_T\transpose 0_{n\times n}X_T+0_{1\times n}X_T+0)}=1.
\end{equation*}
Hence $(X,Y,Z)$ given by (\ref{TheoremX})-(\ref{TheoremZ}) satisfy the FBSDE
(\ref{FBSDEsFWModelA})-(\ref{FBSDEsFWModelAb}).

Second, we prove the uniqueness of the solution. Let $(X, Y,Z)$
be any adapted solution of the FBSDE (\ref{FBSDEsFWModelA})-(\ref{FBSDEsFWModelAb}). Define
\begin{align*} \log\bar{Y}_t&= X_t\transpose
R_2(t)X_t+R_1(t)X_t+R_0(t),\\
\bar{Z}_t&=[X_t\transpose (R_2(t)+R_2\transpose
(t))+R_1(t)]\sigma\bar{Y}_t,
\end{align*}
then 
\begin{equation} 
   \frac{\bar{Z}_t}{\bar{Y}_t}=[X_t\transpose (R_2(t)+R_2\transpose (t))+R_1(t)]\sigma.\label{eq9}
\end{equation}
Apply  It\^{o}'s formula to  $f(t,x)=x\transpose R_2(t)x+R_1(t)x+R_0(t)$.
where $X_t$ is given by (\ref{FBSDEsFWModelA}) to find
\begin{align}
df(t,X_t)&=d\log\bar{Y}_t=  \{ X_t\transpose \dot{R}_2(t)X_t+\dot{R_1}(t)X_t+\dot{R}_0(t) \nonumber\\
&+X\transpose _t(R_2\transpose (t)+R_2(t))AX_t+R_1(t)AX_t+B\transpose (R_2(t)+R_2\transpose (t))X_t+R_1(t)B \nonumber \\
&+[X_t\transpose (R_2\transpose (t)+R_2(t))+R_1(t)]\sigma\frac{Z_t\transpose }{Y_t}
+\trace{\sigma\transpose R_{2}(s) \sigma} \} \ dt
+[X_t\transpose (R_2\transpose (t)+R_2(t))+R_1(t)]\sigma
dW_t^T.\label{eq10}
\end{align}
Substitute equations (\ref{Riccati2}), (\ref{Riccati2b}), (\ref{TheoremR0}) and (\ref{eq9}) into equation
(\ref{eq10}) to find 
\begin{align*} d\log\bar{Y}_t&=\{X_t\transpose
\Gamma X_t+ RX_t+k-\frac{1}{2}\frac{\bar{Z}_t\bar{Z}_t\transpose
}{\bar{Y}_t^2}+\frac{\bar{Z}_t}{\bar{Y}_t}\frac{Z_t\transpose
}{Y_t}\}dt+\frac{\bar{Z}_t}{\bar{Y}_t}dW_t^T
\end{align*}
and
$\log\bar{Y}_T=0.$
So \begin{equation}
\log\bar{Y}_t=-\int\limits_t\limits^T\{X_u\transpose \Gamma X_u+R
X_u+k-\frac{1}{2}\frac{\bar{Z}_u\bar{Z}_u\transpose
}{\bar{Y}_u^2}+\frac{\bar{Z}_u}{\bar{Y}_u}\frac{Z_u\transpose
}{Y_u}\}du-\int\limits_t\limits^T\frac{\bar{Z}_u}{\bar{Y}_u}dW_u^T.\label{eq11}
\end{equation}
Subtract equation (\ref{eq11}) from equation (\ref{eq2}) to find
\begin{align}\label{eq12} \log
Y_t-\log\bar{Y}_t&=-\int\limits_t\limits^T\{\frac{1}{2}\frac{Z_uZ_u\transpose }{Y_u^2}-\frac{\bar{Z}_uZ_u\transpose }{\bar{Y}_uY_u}+\frac{1}{2}\frac{\bar{Z}_u\bar{Z}_u\transpose }{\bar{Y}_u^2}\}du-\int\limits_t\limits^T\{\frac{Z_u}{Y_u}-\frac{\bar{Z}_u}{\bar{Y}_u}\}dW_u^T\nonumber\\
&=-\frac{1}{2}\int\limits_t\limits^T\{(\frac{Z_u}{Y_u}-\frac{\bar{Z}_u}{\bar{Y}_u})(\frac{Z_u\transpose
}{Y_u}-\frac{\bar{Z}_u\transpose}{\bar{Y}_u})\}du-\int\limits_t\limits^T(\frac{Z_u}{Y_u}-\frac{\bar{Z}_u}{\bar{Y}_u})dW_u^T.
\end{align}
Define 
\begin{align*}
\hat{Y}_t=&\log Y_t-\log \bar{Y}_t\\
\hat{Z}_t=&\frac{Z_t}{Y_t}-\frac{\bar{Z}_t}{\bar{Y}_t}.
\end{align*}
Then equation (\ref{eq12}) becomes
\begin{equation}\label{TheoremHatY}
\hat{Y}_t=-\frac{1}{2}\int\limits_t\limits^T\hat{Z}_u\hat{Z}_u\transpose
du-\int\limits_t\limits^T\hat{Z}_udW_u^T.
\end{equation}
By the result of \revised{\citet[Theorem~2.3]{Kobylanski}},
the BSDE (\ref{TheoremHatY})
admits \revised{the} unique adapted solution $(\hat{Y}_t,\hat{Z}_t)=(0,
0_{1\times n})$. So we have $Y_t=\bar{Y}_t$ and $Z_t=\bar{Z}_t$.
This means that any adapted solution $(X, Y, Z)$ of the FBSDE
(\ref{FBSDEsFWModelA})-(\ref{FBSDEsFWModelAb}) must satisfy (\ref{TheoremX}),
(\ref{TheoremY}) and (\ref{TheoremZ}).
\end{proof}

Since  $Y_t=P(t,T)$ a simple corollary of Theorem
\ref{TheoremFBA} gives that the zero coupon bond price is an exponential quadratic function of the factor process.
\begin{corollary}\label{cor3}
If the factor process is given by (\ref{FactorModel}) and the short
rate process is represented by Assumption~\ref{ModelA}, then the zero coupon bond price
has \revised{an} exponential quadratic form,\begin{equation*}
P(t,T)=\text{exp}\{X_t\transpose R_2(t)X_t+R_1(t)X_t+R_0(t)\},
\end{equation*}
where $R_2(t)$, $R_1(t)$ and $R_0(t)$ solve equations
(\ref{Riccati2}), (\ref{Riccati2b}) and (\ref{TheoremR0}), respectively.
\end{corollary}

\begin{remark}
The existence and uniqueness of the solution to Riccati-type differential equations (\ref{Riccati2})-(\ref{Riccati2b}) and (\ref{TheoremR0}) is similar to that of  \citet{gombani2012arbitrage} where the QTSM is characterized in terms of a linear-quadratic control (LQC) problem.  Indeed, much more is known about the solvability of the Riccati equations associated with LQC problems then those arising in ATSMs.  We defer \revised{a detailed} discussion of the \revised{explicit} solvability of the Riccati equations arising in the QTSM problem to the Appendix where we consider more general equations which include (\ref{Riccati2})-(\ref{Riccati2b}),  as well as those arising in our results on futures and forward prices, as special cases.  %
\end{remark}

We next consider a model of risky asset prices that allows us to apply similar techniques to characterize the futures price and forward price of the asset.

\section{Quadratic Price Models}\label{sec:fut-fwd}

Consider a risky asset with price given as an exponential quadratic function of the factors process (\ref{FactorModel}).  This class of price processes allows for the consideration of futures and forward contracts on zero-coupon bonds where the bond price is given as in Corollary~\ref{cor3}.  \revised{
\begin{assumption}\label{ModelA2}
Assume that the risk-neutral dynamics of the factor process are given by equation~(\ref{FactorModel}) and that the price of the risky asset $S$ at time $t\in[0,T]$  is given by $S(t,X_{t})$ where 
\begin{equation}
S(t,x) = \exp{\left( x\transpose a(t) x + b(t) x + c(t)\right) } \label{eq:Sfunc}
\end{equation}
where  $a(t)\in\real^{n\times n}$, $b(t)\in\real^{1\times n}$, and $c(t)\in\real$ are continuous on the interval $[0,T]$ and $a(T)$ is \drevised{negative} semidefinite.  
\end{assumption}
}
Equation~(\ref{eq:Sfunc}) defines a quadratic price model (QPM).  We next extend the results of the previous section and \citet{Cody2009} to consider the futures price and forward price associated with a QPM.

\subsection{Futures Prices}
Consider the $T$-futures price of the risky asset $S$ at time $t\in[0,T]$ defined by
$$G(t,T) = E_{Q}[S(T,X_{T})| \filtration{t}{}]$$
where  $S(t,x)$ is given by equation~(\ref{eq:Sfunc}).
Similar to the results of the previous section and the results of \citet{Cody2009} we may characterize the factor process and futures price 
as the solution to a FBSDE.

Define $Y_{t}=G(t,T)$ so that, by the Martingale Representation Theorem, there exists an $\filtration{t}{}$-adapted process $Z_{t}=[Z_{t}^{1}, \ldots , Z_{t}^{n} ]$ such that 
\begin{equation}
Y_{t}=Y_{0} + \int_{0}^{t}Z_{u} \ dW_{u}. \label{eq:fut4'}
\end{equation}
Note that $Z_{u}$ is an $(1\times n)$-vector valued process.  Therefore,
$$Y_{t} - Y_{T} = -\int_{t}^{T} Z_{u} \  dW_{u}.$$
Since $Y_{T}=G(T,T)=S(T,X_{T})$ we have the following BSDE for the futures price
\begin{equation*}
Y_{t} = S(T,X_{T}) - \int_{t}^{T}Z_{u} dW_{u} . \label{eq:fut6}
\end{equation*}
Take $N(\cdot) = \exp{\left(\int_{0}^{\cdot}r(X_{u})du\right)}G(\cdot,T)$ as num\'eraire and define the measure $P^{G}$ by the Radon-Nikodym derivative
\begin{equation}
\Lambda_{T} = \left. \frac{dP^{G}}{dQ}\right|_{\filtration{T}{}}= e^{-\int_{0}^{T}r(X_{u})du}\frac{N(T)}{N(0)}
=\frac{G(T,T)}{G(0,T)}=\frac{S(T,X_{T})}{G(0,T)}.
\label{eq:46b}
\end{equation}
Define $\Lambda_{t} =  E_{Q}[\Lambda_{T}| \filtration{t}{}]$ and note from equations~(\ref{eq:fut4'}) and (\ref{eq:46b}) that
$$\Lambda_{t} = E_{Q}\bigg[\frac{S(T,X_{T})}{G(0,T)}\bigg| \filtration{t}{}\bigg] = \frac{G(t,T)}{G(0,T)} = \frac{Y_{t}}{Y_{0}}.$$
Dividing equation~(\ref{eq:fut4'}) by $Y_{0}$ gives that the dynamics of $\Lambda_{t}$ are 
$$\Lambda_{t} = 1+ \int_{0}^{t}  \frac{Z_{u}}{Y_{0}}\ dW_{u} = 1+ \int_{0}^{t} \frac{Y_{u}}{Y_{0}} \frac{Z_{u}}{Y_{u}}\ dW_{u}
= 1 + \int_{0}^{t} \Lambda_{u} \frac{Z_{u}}{Y_{u}}\ dW_{u}.$$
Girsanov's Theorem then gives that the process $\{W_{t}^{G}\}_{0\leq t\leq T}$ defined by
$$W_{t}^{G} = W_{t} - \int_{0}^{t} \frac{Z_{u} \transpose}{Y_{u}}\ du$$
is a standard $(\filtration{t}{},P^{G})$-Brownian motion.
Writing the dynamics of $X_{t}$ and $Y_{t}$ under $P^{G}$ we obtain, similar to \citet{Cody2009} the following quadratic FBSDE for the futures price
\begin{align}
X_{t} & = X_{0} + \int_{0}^{t} \left[ AX_{u} + B + \sigma \frac{Z_{u}\transpose}{Y_{u}} \right] du + \int_{0}^{t} \sigma \ dW^{G}_{u} \label{eq:fut15} \\
Y_{t} & = S(T,X_{T}) - \int_{t}^{T} \frac{Z_{u}Z_{u}\transpose}{Y_{u}} du - \int_{t}^{T} Z_{u}\ dW_{u}^{G}. \label{eq:fut16}
\end{align}
While the dynamics of $X$ are Gaussian, and hence a special case of those considered in \citet{hyndman:Gauss,Cody2009}, the dynamics for $Y$ differ due to the exponential quadratic form of the terminal condition.  Nevertheless, following the methodology of \citet{Cody2009} and the previous section we are able to prove the following result which gives an explicit solution to the coupled quadratic FBSDE~(\ref{eq:fut15})-(\ref{eq:fut16}).  The proof, which is independent of the construction of the FBSDE, is similar to the proof of Theorem~\ref{TheoremFBA} and the results of \citet{Cody2009} and is therefore omitted. \revised{We first present sufficient conditions for the solvability of the Riccati equations.}
\revised{\begin{lemma} \label{lemmaB}
If $a(T)\in\real^{n\times n}$ is negative semidefinite the Riccati-type differential equations
\begin{align}
0 &= \frac{d}{dt} R^{G}_{2}(t) + \left( [R^{G}_{2}(t)]\transpose + R^{G}_{2}(t)\right) A  + \frac{1}{2}\left( [R^{G}_{2}(t)]\transpose + R^{G}_{2}(t)\right) \sigma\sigma\transpose \left( [R^{G}_{2} (t)]\transpose + R^{G}_{2}(t) \right), & R_{2}^{G}(T) = a(T)
 \label{eq:fut_ricc_1} \\
0 &= \frac{d}{dt} R^{G}_{1}(t) + R^{G}_{1}(t) \left \{ A +  \sigma\sigma\transpose \left( [R^{G}_{2}(t)]\transpose + R^{G}_{2}(t)\right) \right\} + B\transpose  \left( [R^{G}_{2}(t)]\transpose + R^{G}_{2}(t) \right), & R^{G}_{1}(T)=b(T) 
 \label{eq:fut_ricc_2}
\end{align}
admit unique solutions $R^{G}_{2}(t)\in\real^{n\times n}$ and  $R^{G}_{1}(t)\in\real^{1\times n}$ for all $t\in[0,T]$.
\end{lemma}
\begin{proof}
The result follows, similar to the proof of Lemma~\ref{lemmaA}, if we set $\Upsilon = 0$ and $\Theta =  a(T)$ in equations~(\ref{eq:App_RE1})-(\ref{eq:App_RE2})  of Theorem~\ref{th:Riccati} and $\Psi=0$ and $\theta = b(T)$ in Corollary~\ref{CorollaryA}.
\end{proof}
}
\begin{theorem}\label{th:fut} 
\revised{If $a(T)\in\real^{n\times n}$ is negative semidefinite and $S(t,x)$ is given by equation~(\ref{eq:Sfunc})} the FBSDE~(\ref{eq:fut15})-(\ref{eq:fut16}) has a unique adapted solution $(X,Y,Z)$ given by
\begin{align*}
X_{t} &= X_{0} + \int_{0}^{t} \left\{ \left(A + \sigma\sigma\transpose  \left( [R^{G}_{2} (u)]\transpose + R^{G}_{2}(u)\right)\right) X_{u} 
      + \left( B + \sigma\sigma\transpose [R^{G}_{1}(u)]\transpose \right)\right\} du  + \int_{0}^{t} \sigma \ dW_{u}^{G} \\
Y_{t} &= \exp{\left( X_{t}\transpose R^{G}_{2}(t) X_{t} + R^{G}_{1}(t) X_{t} + R^{G}_{0}(t) \right)} \\
Z_{t} &=  \left[ X_{t} \transpose \left( [R^{G}_{2}(t)]\transpose  + R^{G}_{t}(t)\right) +R^{G}_{1}(t) \right] \sigma \ Y_{t}
\end{align*}
where \revised{$R_{2}^{G}(t)$ is the solution of equation~(\ref{eq:fut_ricc_1}), $R_{1}^{G}(t)$ is the solution of equation~(\ref{eq:fut_ricc_2}), and}
\begin{equation}
R^{G}_{0}(t) = c(T) + \int_{t}^{T} \left( R^{G}_{1}(u)B + \frac{1}{2}R^{G}_{1}(u)\sigma\sigma\transpose [R^{G}_{1}(u)]\transpose 
+ \trace{\sigma\transpose R^{G}_{2}(u) \sigma} 
\right) 
\ du.
\label{eq:fut_ode}
\end{equation}
\end{theorem}
\begin{corollary}
\revised{Under the conditions of Assumption~\ref{ModelA2} t}he futures price has the exponential quadratic form
$$G(t,T) = \exp{\left( X_{t}\transpose R^{G}_{2}(t) X_{t} + R^{G}_{1}(t) X_{t} + R^{G}_{0}(t) \right)}$$ where 
$R^{G}_{2}(t)$ and  $R^{G}_{1}(t)$ are solutions to equations 
(\ref{eq:fut_ricc_1})-(\ref{eq:fut_ricc_2})  and $R^{G}_{0}(t)$ is given by (\ref{eq:fut_ode}).
\end{corollary}

\begin{remark}
Note that the Riccati equations~(\ref{eq:fut_ricc_1})-(\ref{eq:fut_ricc_2}) are similar to those associated with the bond price.  We defer discussion of the existence, uniqueness, and explicit solution of more general equations which include (\ref{eq:fut_ricc_1})-(\ref{eq:fut_ricc_2}) as special cases to the Appendix.
\end{remark}

We next consider the forward price of the risky asset.

\subsection{Forward Prices}

Recall that the $T$-forward price of the risky asset at time $t\in[0,T]$ is
\begin{equation}
F(t,T) = \frac{E_{Q}[ \exp{\left( -\int_{t}^{T}r(X_{u})du \right)}S(T,X_{T}) | \filtration{t}{}]}{P(t,T)} \label{eq:53a}
\end{equation}
where $P(t,T)$ is the price at time $t$ of the zero coupon bond maturing at time $T$.  In the case that the risk free interest rate is deterministic the futures and forward prices are identical.  Therefore, we shall assume that the interest rate is given as in Section~\ref{sec:2} and the risky asset price $S$ is given by equation~(\ref{eq:Sfunc}).  Further, we assume that the risky asset pays a dividend yield (or convenience yield) so that the discounted asset price is not a ${Q}$-martingale which would reduce the numerator of equation~(\ref{eq:53a}) to the risky asset price.

Similar to Section~\ref{sec:2} in the case of the bond price and \citet{Cody2009} in the case of the forward price in an APM we characterize the factor process and the  numerator of equation~(\ref{eq:53a}), which is the risk neutral present value of a forward commitment to deliver one unit of the risky asset at time $T$, in terms of a FBSDE.  Define 
\begin{equation*}
V_{s} = E_{Q}[ \exp{\left( -\int_{0}^{T}r(X_{u})du \right)}S(T,X_{T}) | \filtration{s}{}]
\end{equation*}
and 
\begin{equation*}
H_{s} = \exp{\left( -\int_{0}^{s}r(X_{u})du \right)}.
\end{equation*}
Let $Y_{s} = V_{s} / H_{s}$ and note that $Y_{s}=F(s,T)P(s,T)$.

Since $V_{s}$ is a martingale we have, by the Martingale Representation Theorem, that there exists an adapted process $J_{t}=[J^{1}_{t},\ldots,J^{n}_{t}]$, represented as a $(1\times n)$ vector valued process, such that 
\begin{equation}
V_{s} = V_{0} + \int_{0}^{s}  J_{u} \ dW_{u}. \label{eq:fwd4}
\end{equation}
Apply It\^{o}'s formula to find that $Y_{t}$ satisfies the BSDE
\begin{equation}
Y_{t} = S(T,X_{T}) - \int_{t}^{T}r(X_{u})Y_{u}\ du -\int_{t}^{T} Z_{u}\ dW_{u} \label{eq:fwd6}
\end{equation}
where $Z_{u} = J_{u}/H_{u}$.

Define the risk-neutral measure for the num\'eraire $N(\cdot) = F(\cdot,T)P(\cdot,T)$, denoted by $Q^{F}$,  by the Radon-Nikodym derivative
\begin{equation*}
\Gamma_{T} = \left. \frac{dQ^{F}}{dQ} \right|_{\filtration{T}{}}
= \frac{F(T,T)P(T,T)}{F(0,T)P(0,T)} \exp{\left( -\int_{0}^{T}r(X_{u}) du \right)} = \frac{ S(T,X_{T})}{F(0,T)P(0,T)}\exp{\left( -\int_{0}^{T}r(X_{u}) du \right)}.
\label{eq:fwdRN}
\end{equation*}
Then, with $\Gamma_{t} =  E_{Q}[\Gamma_{T}|\filtration{t}{}]$, we have that $\Gamma_{t}= V_{t}/V_{0}$ and, by equation~(\ref{eq:fwd4}),
\begin{equation*}
\Gamma_{t} = 1 + \int_{0}^{t} \Gamma_{u} \frac{Z_{u}}{Y_{u}} \ dW_{u}.
\end{equation*}
Hence, by Girsanov's theorem,
\begin{equation}
W^{F}_{t} = W_{t} - \int_{0}^{t}\frac{Z_{u}\transpose}{Y_{u}}\ du
\label{eq:fwd10}
\end{equation}
is an $(\filtration{t}{},Q^{F})$-Brownian motion.
Using equation~(\ref{eq:fwd10}) to write the dynamics of $(X,Y)$, given by equations~(\ref{FBSDEsRN})~and~(\ref{eq:fwd6}), under the measure $Q^{F}$ we obtain the following coupled quadratic FBSDE
\begin{align}
X_{t} &= X_{0} + \int_{0}^{t}\left( AX_{u} + B + \sigma \frac{Z_{u}\transpose}{Y_{u}} \right)\ du + \int_{0}^{t} \sigma \ dW^{F}_{u} \label{eq:fwd11} \\
Y_{t} &= S(T,X_{T}) - \int_{t}^{T} \left(r(X_{u})Y_{u} + \frac{Z_{u} Z_{u}\transpose }{Y_{u}} \right) \ du - \int_{t}^{T} Z_{u} \ dW_{u}^{F}. \label{eq:fwd12}
\end{align}
Note that the FBSDE~(\ref{eq:fwd11})-(\ref{eq:fwd12}) is similar to the FBSDE presented in \citet{Cody2009} for APMs except that the volatility dynamics of $X$ are  simpler while the functions $r$ and $S$ are, as functions of $X_{t}$, quadratic and exponential-quadratic functions rather than affine and exponential-affine functions respectively.  Similar to the result of Section~\ref{sec:2} for the case of the bond and \citet{Cody2009} for APMs the following result gives an explicit solution to the quadratic FBSDE~(\ref{eq:fwd11})-(\ref{eq:fwd12}), which is independent of the construction presented, and hence we omit the proof.
\revised{
\begin{lemma}\label{lemmaC}
If $\Gamma\in\real^{n\times n}$ is positive semidefinite and $a(T)\in\real^{n\times n}$ is negative semidefinite then the Riccati-type differential equations
\begin{align}
0 &= \frac{d}{dt} R^{F}_{2}(t) + \left( [R^{F}_{2}(t)]\transpose + R^{F}_{2}(t) \right) A + \frac{1}{2}\left( [R^{F}_{2}(t)]\transpose + R^{F}_{2}(t) \right) 
\sigma\sigma\transpose  \left( [R^{F}_{2}(t)]\transpose + R^{F}_{2}(t) \right) - \Gamma, & R^{F}_{2}(T) = a(T)
\label{eq:63a} \\
0 &= \frac{d}{dt} R^{F}_{1}(t) + R^{F}_{1}(t) A + B\transpose \left( [R^{F}_{2}(t)]\transpose + R^{F}_{2}(t) \right) + R^{F}_{1}(t) \sigma\sigma\transpose \left( [R^{F}_{2}(t)]\transpose + R^{F}_{2}(t) \right) -R, & R^{F}_{1}(T)=b(T)
\label{eq:64a} 
\end{align}
admit unique solutions $R^{F}_{2}(t)\in\real^{n\times n}$ and $R^{F}_{1}(t)\in\real^{1\times n}$ for all $t\in[0,T]$.
\end{lemma}
\begin{proof}
The result follows, similar to the proof of Lemma~\ref{lemmaA}, if we set $\Upsilon = \Gamma$ and $\Theta =  a(T)$ in equations~~(\ref{eq:App_RE1})-(\ref{eq:App_RE2})  of Theorem~\ref{th:Riccati} and $\Psi=R$ and $\theta = b(T)$ in Corollary~\ref{CorollaryA}.
\end{proof}
}
\begin{theorem}\label{th:fwd}
\revised{If $\Gamma\in\real^{n\times n}$ is positive semidefinite, $a(T)\in\real^{n\times n}$ is negative semidefinite, $r(x)$ is given by equation~(\ref{FactorModel}), and $S(t,x)$ is given by equation~(\ref{eq:Sfunc}) then} the FBSDE~(\ref{eq:fwd11})-(\ref{eq:fwd12}) has a unique adapted solution $(X,Y,Z)$ given by
\begin{align}
X_{t} &= X_{0} + \int_{0}^{t} \left\{ \left[ A + \sigma\sigma\transpose \left( [ R^{F}_{2}(u)]\transpose + R^{F}_{2}(u) \right) \right] X_{u} + \left(B+\sigma\sigma\transpose [R^{F}_{1}(u)]\transpose  \right) \right\} du + \int_{0}^{t} \sigma \ dW_{u}^{F} \\
Y_{t} &= \exp{\left( X_{t} R^{F}_{2}(t) X_{t} + R^{F}_{1}(t) X_{t} + R^{F}_{0}(t) \right)}
 \\
Z_{t} &=  \left[ X_{t}\transpose \left( [R^{F}_{2}(t)]\transpose + R^{F}_{2}(t)\right) X_{t} + [R^{F}_{1}(t)] \right] \sigma \ Y_{t}
\end{align}
where $R_{2}^{F}(t)$ is the solution of equation~(\ref{eq:63a}), $R_{1}^{F}(t)$ is the solution of equation~(\ref{eq:64a}), and 
\begin{equation}
{R}^{F}_{0}(t) = c(T) - \int_{t}^{T}\left(k -  R^{F}_{1}(u)B - \frac{1}{2}R^{F}_{1}(u)\sigma\sigma\transpose [R^{F}_{1}(u)]\transpose 
+ \trace{\sigma\transpose [R^{F}_{2}(u)]\transpose \sigma} \right)\ du \ .
\label{eq:65a} 
\end{equation}
\end{theorem}

\begin{corollary}
\revised{Under the conditions of Assumptions~\ref{ModelA} and \ref{ModelA2} t}he forward price is an exponential quadratic function of the factors process
$$F(t,T) = \frac{ \exp{\left(X_{t}\transpose R^{F}_{2}(t) X_{t} + R^{F}_{1}(t)X_{t} + R^{F}_{0}(t) \right)}}{P(t,T)}
= \frac{ \exp{\left(X_{t}\transpose R^{F}_{2}(t) X_{t} + R^{F}_{1}(t)X_{t} + R^{F}_{0}(t) \right)}}{\exp{\left(X_{t}\transpose R_{2}(t) X_{t} + R_{1}(t)X_{t} + R_{0}(t) \right)}}
$$
where 
$R^{F}_{2}(t)$ and  $R^{F}_{1}(t)$ are the solutions to equations (\ref{eq:63a})-(\ref{eq:64a});  $R^{F}_{0}(t)$ is given by equation~(\ref{eq:65a}); and 
$R_{2}(t)$, $R_{1}(t)$, and $R_{0}(t)$ are the solutions to equations
(\ref{Riccati2}), (\ref{Riccati2b}) and (\ref{TheoremR0}).
\end{corollary}

\begin{remark}
Comparing Theorems~\ref{TheoremFBA}, \ref{th:fut}, and \ref{th:fwd} we see that the Riccati type equations (\ref{eq:63a})-(\ref{eq:64a}) for $R^{F}_{2}(t)$ and  $R^{F}_{1}(t)$ and the integral for ${R}^{F}_{0}(t)$  include as special cases the corresponding terms for the bond and futures prices if we make certain parametric restrictions.  This general form is similar to the Riccati type differential equations of linear quadratic control (LQC). Further, these results suggest a correspondence between our results and the results of \citet{gombani2012arbitrage} which are based on LQC.  We discuss the \revised{explicit} solvability of a general Riccati type equation which includes (\ref{eq:63a})-(\ref{eq:64a}) as a special case, and the correspondence to those of LQC in the Appendix.
\end{remark}

In the next section we briefly consider an example of the application to QTSMs of the method of stochastic flows due to \citet{Elliott} which was originally developed in the context of ATSMs. A general formulation of the method of stochastic flows for QTSMs similar to that presented for ATSMs in Section~4 of \citet{Cody2009} can be developed for QTSMs and extended to QPMs similar to the extension for presented in Sections~5.2 and 5.4 of \cite{Cody2009} for APMs.  However, the main motivation for our consideration of the flows method in the context of quadratic models driven by Gaussian factors, in contrast to the case for affine models driven by Gaussian factors which was considered in  \citet{Elliott} and \citet{hyndman:Gauss}, is to show that a measure change is necessary for the method to be effective.  This objective can easily be accomplished by consideration of an example in the one-dimensional case.

\section{Stochastic Flows}\label{sec:flows}

The FBSDE approach to characterizing the bond price in ATSM introduced in \citet{Cody2009}, and extended to QTSMs in this paper, was motivated by the stochastic flow approach introduced by \citet{Elliott}.  The stochastic flow approach  expresses the price at time $t$ of the $T$-maturity zero coupon bond as a function of the value $x$ at time $t$ of the factor process, and this dependence is denoted $P(t,T,x)$.  By taking the derivative of $P(t,T,x)$ with respect to the initial condition, and then using properties of stochastic flows and their Jacobians, it is possible to express $P(t,T,x)$ as an ordinary differential equation (ODE) which illuminates the nature of the functional dependence of the bond price on the factor process.  In \cite{Elliott} when the factors process is Gaussian and the interest rate is an affine function of the factors process the fact that the derivative of the factor process with respect to $x$ is deterministic leads to a linear ODE for $P(t,T,x)$.  In this case it is then immediate that the bond price is an exponential affine function of the factor process.  These results were extended to characterize futures and forward prices on an asset which is an exponential affine function of a Gaussian factor process in \citet{hyndman:Gauss}.  

However, in \citet{Elliott} when the dynamics of the factors are given by an affine process and the interest rate is an affine function of the factors the fact that the derivative with respect to $x$ of the factor process is not deterministic requires a change of measure in order to obtain a linear ordinary differential equation satisfied by the bond price.  The coefficient of this linear ODE is the conditional expectation of a function of the derivative of the factor process.  In the one-dimensional case, corresponding to the CIR model, semi-group properties of the stochastic flow can be used to show that this coefficient is deterministic.  Unfortunately in the case that the factor process is multi-dimensional there was a gap in the proof of a key approximation lemma of \citet{Elliott} which was used to claim that the coefficient of the linear ODE is deterministic (see \cite{Cody2009,hyndman:local}).  The FBSDE approach of \cite{Cody2009} proved, as a corollary to the main results, that the coefficient of the linear ODE for the bond price in the stochastic flows approach is in fact deterministic\revised{.}%

In the case of a Gaussian factor process where the interest rate is a quadratic function of the factor process it is illustrative to consider the stochastic flow approach in the one-dimensional case.  We find that although the factor process is Gaussian, and the derivative of the flow is deterministic, that the change of measure is still necessary in order to obtain an ODE for the bond price.

Suppose  the dynamics of the factor process, $X_{t}$, \revised{are given on the risk-neutral probability space $(\Omega,\mathcal{F},\{\mathcal{F}_t,t \geq 0\},Q)$ by} the one-dimensional model
\begin{equation}
  dX_t=\beta(\alpha-X_t)dt+\sigma dW_t \label{onefactorprocess}
\end{equation}  
and the riskless interest rate $r_t$ is given by the quadratic function of the factor process
\begin{equation}
  r(X_t)=cX_t^2+bX_t+a. \label{oneshortrate} 
\end{equation} 
Write $X_s^{t,x}$ for the solution of
(\ref{onefactorprocess}) started from $x\in \mathbf{R}$ at time
$t\geq 0$. That is, $X_s^{t,x}$ satisfies
\begin{equation*}\label{oneflows}
X_s^{t,x}=x+\int\limits_t\limits^s\beta(\alpha-X_u^{t,x})du+\sigma\int\limits_t\limits^sdW_u,
\quad s\in[t,T].
\end{equation*} 
We refer to $X^{t,x}_{s}$ as the \textit{stochastic flow} associated with the factors process.  \revised{Stochastic flows have been studied extensively by many authors including \citet{MR517235}, \citet{MR620987,MR629977}, and \citet{MR876080,MR1070361}.}  The map $x\rightarrow X^{t,x}_{s}$ is \revised{$Q$-}almost surely differentiable and the derivative satisfies
\begin{equation}
\frac{\partial X_s^{t,x}}{\partial x} = 1 - \beta \int_{t}^{s} \frac{\partial
X_u^{t,x}}{\partial x} du, \quad s \geq t \label{oneDefineD}
\end{equation}
\citet[Theorem 39, p 250]{Protter}. 
Equation~(\ref{oneDefineD}) can be solved independent of $x$.  If $D_{ts}$ satisfies 
\begin{equation}\label{oneDynamicsD}
D_{ts} =1-\beta\int\limits_t\limits^sD_{tu}\ du
\end{equation}
then the unique solution to (\ref{oneDynamicsD}) is the exponential
\begin{equation}
D_{ts}=e^{-\beta(s-t)}. \label{eq:8b}
\end{equation}
That is, $D_{ts} =\frac{\partial X_s^{t,x}}{\partial x}$ which does not depend on $x$ for all $0 \leq t \leq s \leq T$.

By the Markov property of $X_t$,
$%
P(t,T)=P(t,T,X_t),
$%
where\begin{equation}\label{defineP}
P(t,T,x)=E\bigg[\text{exp}\bigg\{-\int\limits_t\limits^T
r(X_u^{t,x})du\bigg\}\bigg].
\end{equation}
Taking the derivative of (\ref{defineP}) with respect to $x$, following \citet{Elliott} and \citet{Cody2007CIR,hyndman:Gauss,Cody2009}, we find
\begin{align}
\frac{\partial P(t,T,x)}{\partial
x}&=E\bigg[\bigg(-\int\limits_t\limits^Tr'(X_u^{t,x})\frac{\partial
X_u^{t,x}}{\partial
x}du\bigg)\text{exp}\bigg(-\int\limits_t\limits^T
r(X_u^{t,x})du\bigg)\bigg]\nonumber\\
&=E\bigg[ L(t,T,x)\text{exp}\bigg(-\int\limits_t\limits^T
r(X_u^{t,x})du\bigg)\bigg],\label{DefOfP}
\end{align}
where
\begin{equation*}\label{defineL}
L(t,T,x)=-\int\limits_t\limits^T (2cX_u^{t,x} +b ) D_{tu} \ du.
\end{equation*} 
We may exchange the order of expectation and
differentiation since $b(x,t) := \beta(\alpha-x)$ and
$\sigma(x,t):= \sigma$ satisfy  linear growth conditions
in $x$ and a global Lipschitz condition, the partial derivatives of
$b(x,t)$ and $\sigma(x,t)$ are continuous and satisfy a polynomial
growth condition, and the function
$\exp{\left(-\int_{t}^{T} r(x)du\right)}$ has two continuous
derivatives satisfying a polynomial growth condition (see
\citet[pp. 117-123]{Friedman}).

In contrast to the work of \citet{Elliott} and \citet{hyndman:Gauss}, in the case of ATSMs driven by Gaussian factors, we may not factor $L(t,T,x)$ from the expectation in equation~(\ref{DefOfP}) as this function depends \revised{on} $X_{u}^{t,x}$. Therefore, even though the model of the instantaneous interest rate given by Assumption~\ref{ModelA}, is driven by Gaussian factors we introduce a change of probability measure as in \citet{Elliott} and \citet{Cody2007CIR,Cody2009}.  
That is, in order formulate equation~(\ref{DefOfP}) as an ODE, we express the conditional expectation under the forward measure.  This choice is motivated by the relationship between a squared Gaussian process and the CIR process previously mentioned.  

Recall Definition~\ref{Numeraire} for the forward measure and note that for any $\mathcal F_T$-measurable random variable $\varphi$
with $E_T|\varphi|<\infty$ a general version of Bayes' Theorem (see, \citet[Lemma 3.5.3]{Karatzas}) gives that
\begin{equation}\label{ChangeOfMeasure}
E_T[\varphi|\mathcal{F}_t]=\Lambda_t^{-1}E[\varphi\Lambda_T|\mathcal{F}_t]
\end{equation}
where $\Lambda_{T}$ is given by equation~(\ref{eq:11}) and from equation~(\ref{Lambdat}) we may write
\begin{align}\label{oneLambda}\Lambda_t =\frac{\exp{\left(-\int\limits_0\limits^t r(X_u)du\right)}P(t,T)}{P(0,T)}.\end{align}
Substitute equations (\ref{eq:11}) and (\ref{oneLambda}) into equation (\ref{ChangeOfMeasure}), noting that $\exp{\left(\int_{0}^{t}r(X_{u})du\right)}$ is $\filtration{t}{}$-measurable, to find
\begin{align}\label{1eq1}
E_T[\varphi|\mathcal{F}_t]P(t,T)&=
E\bigg[\varphi\cdot\text{exp}\bigg\{-\int\limits_t\limits^T
r(X_u^{t,X_t})du\bigg\}\bigg|\mathcal{F}_t\bigg]. 
\end{align}
Therefore, with $\varphi = L(t,T,X_{t})$ in equation~(\ref{1eq1}) we find
\begin{align}\label{oneDiffP}
\frac{\partial P(t,T,x)}{\partial
x}\bigg|_{x=X_t}  &=P(t,T,X_{t})E_T\bigg[L(t,T,X_t)\bigg|\mathcal{F}_t\bigg]
\end{align}
which resembles a linear ODE for the bond price.

In order to solve equation~(\ref{oneDiffP}) we must first consider the conditional expectation's dependence on $X_{t}$.  To derive the dynamics of $X^{t,x}_{u}$ under the forward measure we must characterize the Brownian motion under $Q^{T}$ using Girsanov's Theorem.   As in \citet{Elliott} we may show that
\begin{equation*}
\Lambda_t=\Lambda_0-\int\limits_0\limits^t\Theta_u\Lambda_u dW_u,
\end{equation*}
where
\begin{equation}\label{OneTheta} \Theta_u=-\sigma
E_T\bigg[L(u,T,X_u)\bigg|\mathcal{F}_u\bigg].
\end{equation}
Therefore, by Girsanov's Theorem, the process $W^{T}_{t}$ defined by
\begin{equation} \label{eq:18}
W_t^T=W_t+\int\limits_0\limits^t\Theta_u du
\end{equation} 
is a standard Brownian Motion with respect to the forward
measure $Q^T$. 

Using equations~(\ref{OneTheta})-(\ref{eq:18}) the dynamics of $X_s^{t,x}$ under the forward measure are
\begin{equation}\label{oneXUnderFW}
X_s^{t,x}=x+\int\limits_t\limits^s\{\beta(\alpha-X_v^{t,x})-\sigma\Theta_v\}dv+\sigma\int\limits_t\limits^s
dW_v^T.
\end{equation}
Apply It\^{o}'s product rule to the dynamics of $D_{ts}$
given by (\ref{oneDynamicsD}) and $X_s^{t,x}$ given by
(\ref{oneXUnderFW}).  Then, since $D_{ts}(x)$ is of finite variation, we
have
\begin{align}
X_s^{t,x}D_{ts} &=x+\int\limits_t\limits^s D_{tv} dX_v^{t,x}+\int\limits_t\limits^s X_v^{t,x}d D_{tv} \nonumber\\
&=x+\alpha\beta \int\limits_t\limits^s D_{tv}\ dv-2\beta
\int\limits_t\limits^s X_v^{t,x}D_{tv} \ dv
+\sigma \int\limits_t\limits^s
D_{tv} \ dW_v^T   %
-\sigma^2\int\limits_t\limits^s
D_{tv} \ E_T\bigg[\int\limits_v\limits^T (2cX_{v_1}^{v,x}+b)
D_{vv_1} \bigg|_{x=X_v}dv_1\bigg|\mathcal{F}_v\bigg]dv.\label{1eq6}
\end{align}
Evaluate equation~(\ref{1eq6}) at $x=X_{t}$ and take the $\filtration{t}{}$-conditional expectation under the forward
measure $Q^T$ to find
\begin{align}
&E_T[X_s^{t,X_t}D_{ts} |\mathcal{F}_t]=X_t+ \int\limits_t\limits^s
\left\{ \alpha\beta D_{tv} - 2\beta E_T[X_v^{t,X_t}D_{tv} | \mathcal{F}_t] 
-\sigma^2\int\limits_v\limits^T
E_T[D_{tv} \ E_T[(2cX_{v_1}^{v,X_v}+b)D_{vv_1} dv_1 \ | \ \mathcal{F}_v] \ | \ \mathcal{F}_t]
\right\} \ dv.\label{1eq7}
\end{align}
By the tower property of conditional expectation, since $t\leq v\leq
s\leq T$, the conditional expectation in the double integral of equation (\ref{1eq7}) becomes
\begin{align}
E_T[D_{tv} \ E_T[(2cX_{v_1}^{v,X_v}+b)D_{vv_1} \ | \ \mathcal{F}_v] \ | \ \mathcal{F}_t]%
&= E_T[D_{tv} (2cX_{v_1}^{v,X_v}+b)D_{vv_1} |\mathcal{F}_t] %
= 2c E_T[D_{tv}\ D_{vv_1} \ X_{v_1}^{v,X_v} |\mathcal{F}_t]+ b
D_{tv} \ D_{vv_1}, \label{1eq8}
\end{align}
where we have used the fact that $D_{vv_{1}}$ is deterministic for $t \leq v \leq v_{1} \leq T$.
By the flow property we have, for $t\leq v\leq v_1 \leq T$,
\begin{equation}
X_{v_1}^{v,X_v}=X_{v_1}^{v,X_v^{t,X_t}}=X_{v_1}^{t,X_t}
\end{equation}
and by equation~(\ref{eq:8b}), or the chain rule, we have 
\begin{equation}
D_{tv} \ D_{vv_1}=e^{-\beta(v-t)}e^{-\beta(v_{1}-v)} = e^{-\beta(v_{1}-t)} = D_{tv_1}.\label{flowsD}
\end{equation}
Then, substitute equations (\ref{1eq8}) and (\ref{flowsD}) into equation (\ref{1eq7})
to find
\begin{align}\label{1eq2}
E_T[X_s^{t,X_t}D_{ts} |\mathcal{F}_t]&=X_t+\alpha\beta\int\limits_t\limits^s
D_{tv} \ dv - 2 \beta\int\limits_t\limits^s
E_T[X_v^{t,X_t}D_{tv} \ | \ \mathcal{F}_t]dv  \nonumber\\ &
- b \sigma^2 \int\limits_t\limits^s\int\limits_v\limits^T
D_{tv_1} \ dv_1 \ dv 
 - 2 c \sigma^2  \int\limits_t\limits^s\int\limits_v\limits^T
E_T[X_{v_1}^{t,X_t}D_{tv_1}  \ | \ \mathcal{F}_t]\ dv_1 \ dv.
\end{align}
For $t \leq v \leq T$ define
$\tilde{g}(t,v,x)=E_T[X_v^{t,x}D_{tv} ]$.
Then, by the Markov property of $X_{t}$, 
$
\tilde{g}(t,v,X_{t}) = E_T[X_v^{t,X_{t}} \ D_{tv}\ | \ \filtration{t}{}  ]
$
and, by equation~(\ref{1eq2}), we have that 
\begin{align} \tilde{g}(t,s,X_{t})&= X_{t} +\alpha\beta\int\limits_t\limits^s
e^{-\beta (v-t)} \  dv - 2 \beta \int\limits_t\limits^s \tilde{g}(t,v,X_{t}) \ dv 
- b \sigma^2\int\limits_t\limits^s\int\limits_v\limits^T
e^{-\beta(v_1-t)}\ dv_1 \ dv - 2 c \sigma^2\int\limits_t\limits^s\int\limits_v\limits^T
\tilde{g}(t,v_1,X_{t})\ dv_1 \ dv.\label{1eq9b}
\end{align}
Equation~(\ref{1eq9b}) may be solved by considering, for $x\in \real$, the  nonlinear integral equation 
\begin{align} g(t,s,x )&= x +\alpha\beta\int\limits_t\limits^s
e^{-\beta (v-t)} \  dv - 2 \beta \int\limits_t\limits^s g(t,v,x) \ dv %
- b \sigma^2\int\limits_t\limits^s\int\limits_v\limits^T
e^{-\beta(v_1-t)}\ dv_1 \ dv - 2 c \sigma^2\int\limits_t\limits^s\int\limits_v\limits^T
g(t,v_1,x)\ dv_1 \ dv.\label{1eq9}
\end{align}

Differentiating equation (\ref{1eq9}) with twice respect to $s$, we obtain
the following equivalent ODE
\begin{align}
&g''(t,s,x)  = -\alpha \beta^2 e^{-\beta(s-t)} - 2 \beta g'(t,s,x)+b\sigma^2 e^{-\beta(s-t)}+2c\sigma^2 g(t,s,x)  \label{oneODEs} 
\end{align}
with boundary conditions
\begin{align}
&g(t,t,x )  = x \label{bc1} \\
&g'(t,t,x)  =\alpha\beta-2\beta x - b \sigma^2 \int\limits_t\limits^T
e^{-\beta(v_1-t)} \ dv_1-2c\sigma^2\int\limits_t\limits^T g(t,v_1,x) \ dv_1. \label{bc2}
\end{align} %
The ODE~(\ref{oneODEs}) has general solution
\begin{equation*}\label{oneGSolution}
g(t,s,x)=c_1(x) e^{(-\beta+\sqrt{\beta^2+2c\sigma^2})s}+c_2(x)
e^{(-\beta-\sqrt{\beta^2+2c\sigma^2})s}+\frac{b\sigma^2-\alpha\beta^2}{-\beta^2-2c\sigma^2}\cdot
e^{-\beta(s-t)}
\end{equation*}
for functions $c_{1}(x)$ and $c_{2}(x)$ depending only on $x$.
Applying the boundary conditions (\ref{bc1})-(\ref{bc2}) we find that
\begin{align*}
c_1 (x) &= \frac{  \alpha\beta +
\frac{\sigma^2b\eta^2-2c\sigma^2(b\sigma^2-\alpha\beta^2)}{\beta\eta^2}\cdot(e^{-\beta(T-t)}-1)+ x (-\beta+\eta)e^{(-\beta-\eta)(T-t)}
+\frac{(b\sigma^2-\alpha\beta^2)\beta}{\eta^2}+\frac{(b\sigma^2-\alpha\beta^2)(-\beta+\eta)}{\eta^2}\cdot e^{(-\beta-\eta)(T-t)} }{ e^{(-\beta+\eta)T}
 \Big[\beta+\eta+(-\beta+\eta)e^{2\eta(t-T)} \Big] },\\
c_2(x)&= \frac{ \alpha\beta+
\frac{\sigma^2b\eta^2-2c\sigma^2(b\sigma^2-\alpha\beta^2)}{\beta\eta^2}\cdot(e^{-\beta(T-t)}-1)+x (-\beta-\eta)e^{(-\beta+\eta)(T-t)}
+\frac{(b\sigma^2-\alpha\beta^2)\beta}{\eta^2}+\frac{(b\sigma^2-\alpha\beta^2)(-\beta-\eta)}{\eta^2}\cdot e^{(-\beta+\eta)(T-t)} }{
e^{(-\beta-\eta)T} \Big[\beta-\eta+(-\beta-\eta)e^{2\eta(T-t)}\Big]},
\end{align*}
where $\eta =\sqrt{\beta^2+2c\sigma^2}$.

Therefore, since $\tilde{g}(t,s,X_{t})$ satisfies equation~(\ref{1eq9b}) at $x=X_{t}$ we have by uniqueness of the boundary value problem (\ref{oneODEs})-(\ref{bc2}) that 
\begin{equation}
E_T[X_s^{t,X_t}D_{ts} |\mathcal{F}_t] = g(t,s,X_{t}) \label{eq:35}
\end{equation}
for $0\leq t \leq s \leq T$.  That is, $E_T[X_s^{t,X_t} D_{ts} |\mathcal{F}_t]$ is a deterministic function of $X_t$. 

Consider the conditional expectation in equation~(\ref{oneDiffP}).  By Fubini's Theorem and equation~(\ref{eq:35}) we have
\begin{align}
&E_T\bigg[L(t,T,X_t)\bigg|\mathcal{F}_t\bigg] 
=-\int\limits_t\limits^TE_T\bigg[(2cX_u^{t,X_t}+b)D_{tu} \bigg|\mathcal{F}_t\bigg] \ du %
=-b\int\limits_t\limits^T  D_{tu} \ du 
-2c\int\limits_t\limits^T g(t,u,X_{t})du \nonumber \\
&=\frac{c_1(X_{t})(-\beta-\eta)}{\sigma^2}\cdot(e^{(-\beta+\eta)T}-e^{(-\beta+\eta)t})+\frac{c_2(X_{t})(-\beta+\eta)}{\sigma^2}\cdot(e^{(-\beta-\eta)T}-e^{(-\beta-\eta)t}) %
+\frac{b\eta^2-2c(b\sigma^2-\alpha\beta^2)}{\beta\eta^2}\cdot(e^{-\beta(T-t)}-1). \label{1eq3}
\end{align}
Substitute (\ref{1eq3}) in (\ref{oneDiffP}) and write
\begin{equation} 
\frac{\partial P(t,T,x)}{\partial x}\bigg|_{x=X_t}
=P(t,T,X_{t}) \left( A(\tau)X_t + B(\tau) \right) \label{eq:P-ode}
\end{equation}
where
\begin{align}
A(\tau)&=\frac{2c(e^{2\eta\tau}-1)}{\beta-\eta+(-\beta-\eta)e^{2\eta\tau}},\label{oneA}\\
B(\tau)&= \frac{
\Big[\alpha\beta+\frac{(b\sigma^2-\alpha\beta^2)\beta}{\eta^2}\Big]\Big[(-\beta-\eta)(e^{2\eta\tau}-1)-2\eta\Big] 
 + \frac{(b\sigma^2-\alpha\beta^2)(\beta^2-\eta^2)}{\eta^2}\cdot(e^{2\eta\tau}-1)+2\eta e^{\eta\tau}\cdot\frac{b\sigma^2\eta^2-2c\sigma^2(b\sigma^2-\alpha\beta^2)}{\beta\eta^2}}{
\sigma^2\Big[(\beta+\eta)(e^{2\eta\tau}-1)+2\eta\Big] }
\label{oneB}
\end{align}
where\begin{equation*}\tau =T-t.\end{equation*} 
The ODE~(\ref{eq:P-ode}) has solution
 \begin{equation}\label{1eq4}
P(t,T,X_t)=\text{exp}\Big\{\frac{1}{2}\cdot
A(\tau)X_t^2+B(\tau)X_t+C(T,t)\Big\},
\end{equation}
where $C(T,t)$ is a differentiable function from $\mathbf{R}^2$ to
$\mathbf{R}$. By Feynman-Kac Theorem \cite{Karatzas}, $P(t,T,x)$
defined in (\ref{defineP}) satisfies the Cauchy
problem
\begin{align} 
0 &= \frac{\partial P(t,T,x)}{\partial t} + \beta(\alpha-x) \frac{\partial P(t,T,x)}{\partial x}  
+\frac{1}{2}\sigma^2 \frac{ \partial^2 P(t,T,x)}{(\partial x)^2} - (cx^2+bx+a)P(t,T,x), 
\label{FKTheorem} 
\\
1 &= P(T,T,x). \label{FKTheorem-2} 
\end{align}
Substituting (\ref{1eq4}) in PDE~(\ref{FKTheorem}) and dividing 
by $P(t,T,X_t)$, we have
\begin{align}\label{1eq5}
&\frac{1}{2}X_t^2 \frac{\partial A(\tau)}{\partial t}+ X_t\cdot
\frac{\partial B(\tau)}{\partial t}+\frac{\partial C(T,t)}{\partial
t}+ (\beta\alpha-\beta X_t)[A(\tau)X_t+B(\tau)] +\frac{1}{2}\sigma^2\Big[(A(\tau)X_t+B(\tau))^2+A(\tau)\Big]=cX_t^2+bX_t+a
\end{align}
Comparing the coefficients on both sides of equation (\ref{1eq5}),
obtain  an ODE for $C(T,t)$ 
\begin{align} 
0 &= \frac{\partial C(T,t)}{\partial t}+\beta\alpha B(\tau)+\frac{1}{2}\sigma^2B(\tau)^2+\frac{1}{2}\sigma^2A(\tau)-a   \label{eq:107} \\
0 &= C(T,T). \label{eq:108}
\end{align} 
Solving the ODE~(\ref{eq:107})-(\ref{eq:108}), and denoting  
$C(\tau)=C(T,t)$, 
we have
\begin{align}\label{oneC} C(\tau)&=
\bigg[\frac{b^2\sigma^2-2\alpha\beta^2b-2\alpha^2\beta^2c}{2\eta^2}\bigg]\tau+\frac{1}{2}\log\bigg\{\frac{2\eta
e^{(\eta+\beta)\tau}}{(\beta+\eta)e^{2\eta\tau}+\eta-\beta}\bigg\}\nonumber\\
&+\frac{2c(b\sigma^2-\alpha\beta^2)^2+2\beta(b+2c\alpha)(b\sigma^2-\alpha\beta^2)(\beta+\eta)e^{\eta\tau}-\beta^2\sigma^2(b+2c\alpha)^2}{\eta^3(\beta+\eta)\Big[(\beta+\eta)e^{2\eta\tau}+\eta-\beta\Big]}\nonumber\\
&+\frac{\beta^2\sigma^2(b+2c\alpha)^2-2c(b\sigma^2-\alpha\beta^2)^2-2\beta(b+2c\alpha)(b\sigma^2-\alpha\beta^2)(\beta+\eta)}{2\eta^4(\beta+\eta)}-a\tau.
\end{align}
Summarizing the material of this section we have the following result which shows that
$P(t,T,X_{t})$ is an exponential quadratic function of the factor process. 
\begin{theorem} For $t\in[0,T]$ and for all $x\in\mathbf R$, 
\begin{equation*} 
P(t,T,X_{t})=\exp{\left(\frac{1}{2}\cdot A(\tau)X_{t}^2+B(\tau)X_{t}+C(\tau)\right)},
\end{equation*}
where $A(\tau)$, $B(\tau)$ and $C(\tau)$ is given by (\ref{oneA}),
(\ref{oneB}) and (\ref{oneC}), respectively.
\end{theorem}

The Markov property, $P(t,T)=P(t,T,X_{t})$ gives the following characterization of the zero coupon bond price.
\begin{corollary}\label{cor:flow-1d-bond-price}
If the factor process is given by (\ref{onefactorprocess}) and the
short rate is represented by the function (\ref{oneshortrate}), the
zero-coupon bond price
is\begin{equation*}P(t,T)=\text{exp}\Big\{\frac{1}{2}\cdot
A(\tau)X_t^2+B(\tau)X_t+C(T,t)\Big\},
\end{equation*}where $A(\tau)$, $B(\tau)$ and $C(\tau)$ is given by (\ref{oneA}),
(\ref{oneB}) and (\ref{oneC}), respectively.
\end{corollary}
Corollary~\ref{cor:flow-1d-bond-price} agrees with the results in \citet[pp. 487-488]{Nawalkha}.

In higher dimensional cases the stochastic flow method
requires the addition of some parametric restrictions to the model
similar to the ATSM case which was studied by \cite{Martino}.
Nevertheless, the example considered in this section illustrates the 
similarities of the stochastic flow method in the one-dimensional QTSM with 
Gaussian factor process to the one-dimensional ATSM model and provides further motivation for 
the necessity of the change of measure.

\section{Conclusions}\label{sec:conc}

In this paper we have extended the FBSDE approach, introduced in \cite{Cody2009} in the context of affine term structure models, to quadratic term structure models (QTSM) where the factor process is Gaussian and the riskless interest rate is a quadratic functional of the factor process.  After characterizing the factor process and the bond price in terms of a coupled quadratic FBSDE under the forward measure we prove the existence and uniqueness of the FBSDE and provide an explicit solution.  This approach provides new examples of quadratic FBSDEs with explicit solutions.   We extend the FBSDE approach to consider the futures price and forward price of a risky asset with spot price given by an exponential quadratic functional of the factor process, which we term a quadratic price model.  Our results are motivated, as in the ATSM case, by the stochastic flows approach of \citet{Elliott}, and we briefly consider the one-dimensional QTSM in order to illustrate why the change of measure technique is necessary even though the factors are Gaussian.  The results of this paper can be easily extended to consider a Gaussian factor model with time dependent coefficients.

\appendix
\section{Appendix}

We consider the solvability of a non-symmetric Riccati-type matrix differential equation which includes as 
special cases those necessary for the solvability of the FBSDEs given in Theorems~\ref{TheoremFBA}, \ref{th:fut}, and \ref{th:fwd},
as well as their corollaries, which characterize the bond, futures, and forward prices.  While the result presented in this appendix is more general than is required for our purposes we believe it is of independent interest as an example of a solvable non-symmetric Riccati-type equation.

The method of proof involves decomposing the non-symmetric Riccati-type differential equation into the sum of the classical symmetric matrix Riccati equation of linear quadratic optimal control (LQC), which is employed by \citet{gombani2012arbitrage}, and a skew-symmetric matrix.   For further details on the Riccati equations of LQC  see, for example, \citet{MR0335000} and \citet{MR0364275}.  Comprehensive information on Riccati equations can be found in \citet{MR1367089} and \citet{MR1997753}.

We relate our results to those presented in \citet{gombani2012arbitrage} for the bond price in a QTSM.  However, our approach and model parameterization is different from \citet{gombani2012arbitrage} so we obtain slightly different results.  Nevertheless, there is a strong relationship as the next result shows.

\begin{theorem}\label{th:Riccati}
Consider the general $(n\times n)$-matrix Riccati-type differential equation for $R_{2}(t)$
\begin{align}
& \frac{dR_{2}(t)}{dt} + ( [R_{2}(t)]\transpose + R_{2}(t) )A + \frac{1}{2}( [R_{2}(t)]\transpose + R_{2}(t) ) \sigma\sigma\transpose 
( [R_{2}(t)]\transpose + R_{2}(t) ) - \Upsilon = 0  \label{eq:App_RE1} \\
& R_{2}(T) =   \Theta \label{eq:App_RE2}
\end{align}
where $(\Upsilon+\Upsilon\transpose)$ is positive semidefinite and $(\Theta + \Theta\transpose)$ is negative semidefinite.  Then 
\begin{list}{(\roman{enumi})}{\usecounter{enumi}\addtolength{\itemsep}{-0.5\baselineskip}\setlength{\topsep}{0.0pc}\setlength{\partopsep}{0.0pc}\setlength{\listparindent}{0.25in}\setlength{\parsep}{0.0pc}}
\item  a solution $R_{2}(t)$ to (\ref{eq:App_RE1})-(\ref{eq:App_RE2}) always exists on the entire interval $[0,T]$ and can be expressed as
\begin{equation}
R_{2}(t) = - (U(t) + V(t))
\label{eq:R2_decomp}
\end{equation}
where $U(t)$ satisfies
\begin{align}
& \frac{d}{dt} U(t) +  U(t) A + A \transpose U(t) - 2 U(t) \sigma\sigma\transpose U(t) + Q = 0 \label{eq:gb_re1} \\
& U(T) =  C_{1} \label{eq:gb_re2}
\end{align}
with $Q= \frac{1}{2}(\Upsilon + \Upsilon\transpose)$ and $C_{1}=-\frac{1}{2}( \Theta + \Theta\transpose)$ and
$V(t)$ is given by
\begin{equation}
V(t) = \tilde{C}_{1} + \int_{t}^{T}[U(s)A - A\transpose U(s) + \tilde{Q} ]\ ds
\label{eq:V_inteq}
\end{equation}
where $\tilde{Q} = \frac{1}{2}(\Upsilon - \Upsilon\transpose)$ and $\tilde{C}_{1} = -\frac{1}{2}(\Theta - \Theta\transpose)$.  
\item
A solution to (\ref{eq:gb_re1})-(\ref{eq:gb_re2}) always exists in the interval $[0,T]$ and it can be expressed as  
$U(t) = Y(t)X(t)^{-1}$ where
$X$ and $Y$ satisfy the linear differential equation
\begin{equation}
\frac{d}{dt}\left[ \begin{array}{c} {X}(t) \\ {Y}(t) \end{array} \right]
=  \left[ \begin{array}{cc} A & -2\sigma\sigma\transpose \\
- Q & -A\transpose 
\end{array} \right] \left[ \begin{array}{c} {X}(t) \\ {Y}(t) \end{array} \right]
\quad \quad \left[ \begin{array}{c} {X}(T) \\ {Y}(T) \end{array} \right]
= \left[ \begin{array}{c} I  \\ C_{1} \end{array} \right].
\label{eq:App_ODEsys}
\end{equation}
Moreover, if $\Phi(t,s)$ denotes the fundamental solution (transition matrix) associated with 
\begin{equation}
\frac{d}{dt}x(t) = [A - 2\sigma\sigma\transpose U(t)] x(t)
\label{eq:trans1}
\end{equation}
then $X$ and $Y$ admit the following interpretation
$$X(t) = \Phi(t,T)$$
and 
$$ Y(t) = U(t)\Phi(t,T).$$
\end{list}
\end{theorem}
\begin{proof}
The proof of \textit{(ii)} follows immediately as a special case of \citet[Theorem~B.1]{gombani2012arbitrage}.  To prove \textit{(i)} first suppose $R_{2}(t)$ is a solution to (\ref{eq:App_RE1})-(\ref{eq:App_RE2}) on the interval $[0,T]$ and define the symmetric matrix 
\begin{equation}
U(t) = - \frac{1}{2}(R_{2}(t) +  [R_{2}(t)]\transpose ). \label{eq:Usym}
\end{equation}
Differentiating equation~(\ref{eq:Usym}) with respect to $t$ we have, by (\ref{eq:App_RE1}),  that $U(t)$ satisfies 
\begin{align*}
\frac{d}{dt}U(t) &= -\frac{1}{2}\left( \frac{d}{dt}R_{2}(t) + \frac{d}{dt}[R_{2}(t)]\transpose\right) \\
&= \frac{1}{2}\left\{ ( [R_{2}(t)]\transpose + R_{2}(t) )A + \frac{1}{2} ( [R_{2}(t)]\transpose + R_{2}(t) ) \sigma\sigma\transpose ( [R_{2}(t)]\transpose + R_{2}(t) ) - \Upsilon \right\}  \\
&+ \frac{1}{2}\left\{A\transpose ( [R_{2}(t)]\transpose + R_{2}(t) ) 
+ \frac{1}{2}( [R_{2}(t)]\transpose + R_{2}(t) )\sigma\sigma\transpose ( [R_{2}(t)]\transpose + R_{2}(t) ) - \Upsilon\transpose \right\} 
\\
&= -U(t)A - A\transpose U(t) + 2U(t)\sigma\sigma\transpose U(t) - \frac{1}{2}(\Upsilon + \Upsilon\transpose )
\end{align*}
which gives (\ref{eq:gb_re1}).  Evaluating (\ref{eq:Usym}) at $t=T$ and applying  (\ref{eq:App_RE2}) gives 
$$U(T) = -\frac{1}{2}(R_{2}(T) + [R_{2}(T)]\transpose) = -\frac{1}{2}[\Theta + \Theta\transpose] = C_{1}$$
which is (\ref{eq:gb_re2}). Next, define the skew-symmetric matrix $V(t)$ by
\begin{equation}
V(t) = -\frac{1}{2}(R_{2}(t) - [R_{2}(t)]\transpose)
\label{eq:Vskew}
\end{equation}
and differentiate with respect to $t$ to find that $V(t)$ satisfies
\begin{align}
\frac{d}{dt} V(t) &= -\frac{1}{2}\left( \frac{d}{dt} R_{2}(t) - \frac{d}{dt}[R_{2}(t)]\transpose \right) \nonumber \\
& = -\frac{1}{2} \left\{ -( [R_{2}(t)]\transpose + R_{2}(t) )A - \frac{1}{2}( [R_{2}(t)]\transpose + R_{2}(t) )\sigma\sigma\transpose ( [R_{2}(t)]\transpose + R_{2}(t) ) + \Upsilon  \right. \nonumber \\
&+ \left. A\transpose ( [R_{2}(t)]\transpose + R_{2}(t) ) + \frac{1}{2} ( [R_{2}(t)]\transpose + R_{2}(t) )\sigma\sigma\transpose ( [R_{2}(t)]\transpose + R_{2}(t) ) - \Upsilon\transpose \right\} \nonumber \\
&= - \left\{ U(t)A - A\transpose U(t) + \tilde{Q} \right\}.
 \label{eq:Vde}
\end{align}
Evaluate (\ref{eq:Vskew}) at $t=T$ and apply (\ref{eq:App_RE2}) to find
\begin{equation}
V(T) = -\frac{1}{2}(R_{2}(T) - [R_{2}(T)]\transpose) = -\frac{1}{2}(\Theta - \Theta\transpose) = \tilde{C}_{1}.  
\label{eq:Vtc}
\end{equation}
Therefore, solving (\ref{eq:Vde})-(\ref{eq:Vtc}) gives that $V(t)$ satisfies (\ref{eq:V_inteq}).

Conversely, suppose that $R_{2}(t)$ is defined by (\ref{eq:R2_decomp})-(\ref{eq:V_inteq}). Note that any solution to  (\ref{eq:gb_re1})-(\ref{eq:gb_re2}) is symmetric and $V(t)$ given by (\ref{eq:V_inteq}) is skew-symmetric.  Therefore, by the uniqueness of the decomposition of a square matrix into the sum of symmetric and skew-symmetric matrices, we must have $U(t) = -\frac{1}{2}(R_{2}(t) +  [R_{2}(t)]\transpose )$ and $V(t) = -\frac{1}{2}(R_{2}(t) - [R_{2}(t)]\transpose)$. Then $R_{2}(t)$ satisfies
\begin{align*}
\frac{d}{dt}R_{2}(t) &= -\frac{d}{dt}U(t) - \frac{d}{dt}V(t) = 
-[-U(t)A - A\transpose U(t) + 2 U(t)\sigma\sigma\transpose U(t) - Q] - [ -U(t)A + A\transpose U(t) - \tilde{Q} ]   \\
&= 2 U(t)A -2 U(t)\sigma\sigma\transpose U(t) + Q + \tilde{Q} \\
&=  2\left[ -\frac{1}{2}(R_{2}(t)+[R_{2}(t)]\transpose) \right] A - 2\left[ -\frac{1}{2}(R_{2}(t)+[R_{2}(t)]\transpose) \right] \sigma\sigma\transpose  \left[ -\frac{1}{2}(R_{2}(t)+[R_{2}(t)]\transpose) \right] + Q + \tilde{Q} \\
&= - (R_{2}(t)+[R_{2}(t)]\transpose)A - \frac{1}{2} (R_{2}(t)+[R_{2}(t)]\transpose) \sigma \sigma\transpose (R_{2}(t)+[R_{2}(t)]\transpose) + \Upsilon
\end{align*}
which is (\ref{eq:App_RE1}) as desired.  Finally, evaluating  (\ref{eq:R2_decomp}) at $t=T$ gives
$$R_{2}(T) = - (U(T) + V(T)) = - C_{1} - \tilde{C}_{1} =  \frac{1}{2}(\Theta + \Theta\transpose) + \frac{1}{2}(\Theta - \Theta\transpose) = \Theta$$
so that  (\ref{eq:App_RE2}) is satisfied.
\end{proof}

\begin{remark}\label{rem:App}
Define the Hamiltonian associated with (\ref{eq:App_ODEsys})
$$ H = \left[ \begin{array}{cc} A & -2\sigma\sigma\transpose \\
- Q & -A\transpose 
\end{array} \right].
$$
Since $H$ is constant there is an explicit representation
$$
\left[ \begin{array}{c} {X}(t) \\ {Y}(t) \end{array} \right]
= e^{H(t-T)}\left[ \begin{array}{c} I \\ C_{1} \end{array} \right].
$$
Therefore, by Theorem~\ref{th:Riccati}, $R_{2}(t)$ has an explicit solution on the interval $[0,T]$.
\end{remark}

\begin{remark}
In the case of the bond price $\Upsilon=\Gamma$ and $\Theta = 0$.  Therefore, by Assumption~\ref{ModelA}, since $\Gamma$ is positive semidefinite, and implicitly symmetric, we have that $\tilde{Q}=0$.  Further, since $\tilde{C}_{1}=0$ we have that, in the case of the bond price, equation (\ref{eq:V_inteq}) simplifies to 
$$V(t) = \int_{t}^{T} [U(s) A - A\transpose U(s)] \ ds.$$
In order to provide an exact correspondence with the results of \citet{gombani2012arbitrage} it seems as though we should have 
$V(t)=0$ for all $t\in [0,T]$ which is clearly true in the one-dimensional case where $A$ is a scalar, however, in the  multi-dimensional case the result is not obvious.  Nevertheless, owing to the different parameterization and methods of this paper and those of \citet{gombani2012arbitrage}, the difference should be no more alarming than that between the Riccati equations obtained by substitution into the term-structure PDE in \citet[equations~(2.4)-(2.6)]{gombani2012arbitrage} and those obtained from the LQC approach in \citet[equation~(3.12)]{gombani2012arbitrage}.  %
\end{remark}

We next consider the solution of the associated differential equation for $R_{1}(t)$ special cases of which appear in Theorems~\ref{TheoremFBA}, \ref{th:fut}, and \ref{th:fwd} and their corollaries for the bond, futures, and forward prices.
\begin{corollary}\label{CorollaryA}
Let $R_{1}(t)$ be the solution to 
\begin{align}
& 0 = \frac{d}{dt}R_{1}(t) + R_{1}(t) \left\{ A  +  \sigma\sigma\transpose \left( [R_{2}(t)]\transpose + R_{2}(t) \right) \right\} +  B \transpose \left( [R_{2}(t)]\transpose + R_{2}(t) \right) - \Psi \label{ode:R1} \\
& R_{1}(T) = \theta . \label{term:R1}
\end{align}
Then $R_{1}(t)$ can be written as
\begin{equation}
R_{1}(t) = \left( \theta - \int_{t}^{T} [ 2 B\transpose Y(s) + \Psi X(s)] \ ds \right) [X(t)]^{-1}.
\label{eq:R1-sol}
\end{equation}
\end{corollary}
\begin{proof}
The result may be verified by simple differentiation, however, we provide the construction.  Taking the transpose of equation (\ref{ode:R1}) and applying the fact that $U(t) = -\frac{1}{2}(R_{2}(t) + [R_{2}(t)]\transpose)$ we may write
\begin{equation}
\frac{d}{dt}[R_{1}(t)]\transpose = ( 2U(t) \sigma\sigma\transpose - A\transpose ) [R_{1}(t)]\transpose + K(t) \label{eq:R1-2}
\end{equation}
where $K(t) = (2 U(t) B + \Psi\transpose )$.  Let $\Psi(t)$ be the $(n\times n)$-matrix whose columns are the vectors which form a fundamental set of solutions to 
\begin{equation}
\frac{d p(t)}{dt} = (2 U(t)\sigma\sigma\transpose - A\transpose) p(t) \label{eq:adjoint}
\end{equation}
and note that equation~(\ref{eq:adjoint}) is the adjoint equation corresponding to equation~(\ref{eq:trans1}). Then, by variation of parameters, the solution to (\ref{eq:R1-2}) with terminal condition~(\ref{term:R1})
is 
$$
[R_{1}(t)]\transpose = \Psi(t) \left( [\Psi(T)]^{-1}\theta\transpose - \int_{t}^{T} [\Psi(s)]^{-1} K(s) \ ds \right).
$$
Define
$$
\Psi(t,s) = \Psi(t)\Psi^{-1}(s)
$$
and note that $\Psi(t,s)$ is the transition matrix of (\ref{eq:adjoint}).  We also have that, since
$\Phi(t,s)$ is the transition matrix of (\ref{eq:trans1}) and (\ref{eq:adjoint}) is the adjoint, that
$$\Psi(t,s) = \left([\Phi(t,s)]^{-1}\right)\transpose
= [\Phi(s,t)]\transpose.
$$
Therefore, from (\ref{eq:R1-2}), we may write
\begin{align}
R_{1}(t) &= \theta [\Psi(t,T)]\transpose - \int_{t}^{T} [K(s)]\transpose [\Psi(t,s)]\transpose \ ds \nonumber \\
&= \theta \Phi(T,t) - \int_{t}^{T} [2 B\transpose U(s) + \Psi ]\transpose \Phi(s,t) \ ds.
\label{eq:R1-p4}
\end{align}
From Theorem~\ref{th:Riccati} we have $U(s) = Y(s)[X(s)]^{-1}$, $\Phi(s,t)=X(s)[X(t)]^{-1}$, and $X(T)=I$
so the result follows from equation~(\ref{eq:R1-p4}).
\end{proof}

\begin{remark}
Since, as noted in Remark~\ref{rem:App}, there is an explicit representation for $X(t)$ and $Y(t)$  the integral~(\ref{eq:R1-sol}) can be computed explicitly.
\end{remark}

\vspace{5mm}

\noindent
\textbf{Acknowledgements}
\\ \noindent
C. Hyndman would like to acknowledge the financial support of the Natural Sciences and Engineering Research Council of Canada (NSERC) and the Institut de finance math\'{e}matique de Montr\'{e}al (IFM2).

\nocite{alan}

\bibliographystyle{abbrvnat} %
\bibliography{qtsm_bib}

\end{document}